\documentclass[a4paper,UKenglish,cleveref, autoref, thm-restate]{lipics-v2021}

\title{Multilevel Hypergraph Partitioning with Vertex Weights Revisited}

\titlerunning{Multilevel Hypergraph Partitioning with Vertex Weights Revisited}

\author{Tobias Heuer}{Karlsruhe Institute of Technology, Karlsruhe, Germany}{tobias.heuer@kit.edu}{}{}

\author{Nikolai Maas}{Karlsruhe Institute of Technology, Karlsruhe, Germany}{nikolai.maas@student.kit.edu}{}{}

\author{Sebastian Schlag}{Karlsruhe Institute of Technology, Karlsruhe, Germany}{research@sebastianschlag.de}{}{}

\authorrunning{T. Heuer, N. Maas and S. Schlag}

\Copyright{Tobias Heuer, Nikolai Maas and Sebastian Schlag}



\category{}

\relatedversion{}

\supplement{Source Code: \url{https://github.com/kahypar/kahypar} \\ Benchmark Set \& Experimental Results: \url{http://algo2.iti.kit.edu/heuer/sea21/}}

\nolinenumbers

\hideLIPIcs

\EventEditors{John Q. Open and Joan R. Access}
\EventNoEds{2}
\EventLongTitle{42nd Conference on Very Important Topics (CVIT 2016)}
\EventShortTitle{SEA 2021}
\EventAcronym{SEA}
\EventYear{2021}
\EventDate{December 24--27, 2016}
\EventLocation{Little Whinging, United Kingdom}
\EventLogo{}
\SeriesVolume{42}
\ArticleNo{23}
\usepackage{amsmath,amsfonts,amssymb,pifont,marvosym}
\usepackage{multicol}
\usepackage{thmtools,thm-restate}
\usepackage{nicefrac}
\usepackage{booktabs}
\usepackage{multirow}
\usepackage{float}
\usepackage{tikz}
\usepackage{graphicx}
\usepackage[section]{placeins}
\usepackage[ruled,vlined,linesnumbered,norelsize]{algorithm2e}
\DontPrintSemicolon

\SetKwSty{texttt}
\SetCommentSty{emph}

\usepackage{pgfplots}
\usepgfplotslibrary{external}
\tikzsetexternalprefix{stats/}

\newcommand{\externalizedfigure}[2]{
  \includegraphics{experiments/paper-figure#1}
}

\newcommand{\definitionname}[1]{\emph{(#1)}\textbf{.}}

\newcommand{\Oh}[1]{\ensuremath{\mathcal{O}(#1)}}

\newcommand{\Partition}{\ensuremath{\mathrm{\Pi}}}%
\newcommand{\DefPartition}[1]{\ensuremath{\mathrm{\Pi}_{\scriptstyle #1}}}%
\newcommand{\Prepacking}{\ensuremath{\mathrm{\Psi}}}%
\newcommand{\DefPrepacking}[1]{\ensuremath{\mathrm{\Psi}_{\scriptstyle #1}}}%

\newcommand{\subhypergraph}[1]{\ensuremath{H_{#1}}}
\newcommand{\fixedvertices}{\ensuremath{P}}

\newcommand{\balancedconstraint}[1]{\ensuremath{L_{#1}}}
\newcommand{\specialbalancedconstraint}[2]{\ensuremath{L_{#1}^{\scriptstyle #2}}}

\newcommand{\subpartition}[2]{\ensuremath{V_{#1}^{\scriptscriptstyle (#2)}}}

\newcommand{\afdbound}[2]{\ensuremath{h_{#1}(#2)}}

\newcommand{\loadbalancingweight}{\ensuremath{c}}

\newcommand{\nonprepackedvertices}{\ensuremath{O}}
\newcommand{\heaviestvertices}[1]{\ensuremath{\nonprepackedvertices_{#1}}}

\newcommand{\opt}{\ensuremath{\mathrm{OPT}}}
\newcommand{\lpt}{\ensuremath{\mathrm{LPT}}}

\newcommand{\Partitioner}[1]{\texttt{#1}}
\newcommand{\InstanceType}[1]{\textsc{#1}}

\definecolor{fuchsiapink}{rgb}{1.0, 0.47, 1.0}

\newcommand{\plusplus}{\texttt{++}}
\newcommand{\Cpp}[1]{C\plusplus#1}
\newcommand{\gpp}[1]{g\plusplus#1}

\newcommand{\splitatcommas}[1]{%
  \begingroup
  \begingroup\lccode`~=`, \lowercase{\endgroup
    \edef~{\mathchar\the\mathcode`, \penalty0 \noexpand\hspace{0pt plus 1em}}%
  }\mathcode`,="8000 #1%
  \endgroup
}

\begin{document}

\maketitle

\begin{abstract}
The balanced hypergraph partitioning problem (HGP) is to partition the vertex set of a hypergraph into $k$
disjoint blocks of bounded weight, while minimizing an objective function defined on the hyperedges.
Whereas real-world applications often use vertex and edge weights to accurately model the underlying problem,
the HGP research community commonly works with unweighted instances.

In this paper, we argue that, in the presence of vertex weights, current balance constraint definitions
 either yield infeasible partitioning problems
or allow unnecessarily large imbalances and propose a new definition
that overcomes these problems. We show that state-of-the-art hypergraph partitioners often struggle considerably
with weighted instances and tight balance constraints (even with our new balance definition). Thus, we present a
recursive-bipartitioning technique that is able to reliably compute balanced (and hence feasible) solutions. The
proposed method balances the partition by pre-assigning a small subset of the heaviest vertices to the two
blocks of each bipartition (using an algorithm originally developed for the job scheduling problem) and optimizes
the actual partitioning objective on the remaining vertices.
We integrate our algorithm into the multilevel hypergraph partitioner \Partitioner{KaHyPar} and show that our approach
is able to compute balanced partitions of high quality on a diverse set of benchmark instances.
\end{abstract}

\section{Introduction}
\label{sec:intro}

Hypergraphs are a generalization of graphs where each hyperedge can connect more than two vertices.
The $k$-way hypergraph partitioning problem (HGP) asks for a partition of the vertex set into $k$ disjoint
blocks, while minimizing an objective function defined on the hyperedges. Additionally, a balance constraint
requires that the weight of each block is smaller than or equal to a predefined upper bound
(most often $\balancedconstraint{k} := (1+\varepsilon)\lceil \frac{c(V)}{k} \rceil$ for some parameter $\varepsilon$,
where $c(V)$ is the sum of all vertex weights).
The hypergraph partitioning problem is NP-hard~\cite{Lengauer:1990} and it is even NP-hard to find good
approximations~\cite{DBLP:journals/ipl/BuiJ92}.
The most commonly used heuristic to solve HGP
in practice is the multilevel paradigm~\cite{KaHyPar-K, ccatalyurek1996decomposing, DBLP:conf/dac/KarypisAKS97}
which consists of three phases:
First, the hypergraph is \emph{coarsened} to obtain a hierarchy of smaller hypergraphs.
After an \emph{initial partitioning} algorithm is applied to the smallest hypergraph, \emph{coarsening} is undone, and,
at each level, \emph{refinement} algorithms are used to improve the quality of the solution.

The two most prominent application areas of HGP are
very large scale integration (VLSI) design~\cite{DAlpert, DBLP:conf/dac/KarypisAKS97} and parallel computation
of the sparse matrix-vector product~\cite{ccatalyurek1996decomposing}.
In the former, HGP is used to divide a circuit into two or more blocks such that the number of external wires interconnecting circuit elements in different blocks is minimized. In this setting, each vertex is associated with a weight equal to the area of the respective circuit element~\cite{ISPD98} and tightly-balanced partitions minimize the total area required by the physical circuit~\cite{DBLP:conf/iccad/DuttT97}.
In the latter, HGP is used to optimize the communication volume for parallel computations of
sparse matrix-vector products~\cite{ccatalyurek1996decomposing}. In the simplest hypergraph model, vertices correspond to rows
and hyperedges to columns of the matrix (or vice versa) and a partition of the hypergraphs yields an assignment of matrix entries to processors~\cite{ccatalyurek1996decomposing}.
The work of a processor (which can be measured in terms of the number of non-zero entries~\cite{bisseling2012two}) is integrated into
the model by assigning each vertex a weight equal to its degree~\cite{ccatalyurek1996decomposing}.
Tightly-balanced partitions hence ensure that the work is distributed evenly among the processors.

Despite the importance of weighted instances for real-world applications, the HGP research community mainly uses
unweighted hypergraphs in experimental evaluations~\cite{KAHYPAR-DIS}. The main rationale hereby being that even unweighted instances
become weighted implicitly due to vertex contractions during the coarsening phase. Many partitioners therefore incorporate techniques that prevent the formation of heavy vertices~\cite{DBLP:journals/tpds/CatalyurekA99, Hauck:1996:MS:238604, KaHyPar-CA} during coarsening to facilitate finding
a feasible solution during the initial partitioning phase~\cite{KAHYPAR-DIS}.
However, in practice, many weighted hypergraphs derived from real-world applications already contain heavy vertices -- rendering
the mitigation strategies of today's multilevel hypergraph partitioners ineffective. The popular ISPD98 VLSI benchmark set~\cite{ISPD98}, for example,
includes instances in which vertices can weigh up to $10\%$ of the total weight of the hypergraph.

\subparagraph*{Contributions and Outline}
After introducing basic notation in Section~\ref{sec:preliminaries}
and presenting related work in Section~\ref{sec:related_work}, we first
formulate an alternative balance constraint definition in Section~\ref{sec:balance_constraint} that
 overcomes some drawbacks of existing definitions in presence of vertex weights.
 In Section~\ref{sec:balanced_partitioning}, we then present an algorithm that enables partitioners based on the
 recursive bipartitioning (RB) paradigm to reliably compute balanced partitions for weighted hypergraphs.
 Our approach is based on the observation that usually only a small subset of the heaviest vertices
 is \emph{critical} to satisfy the balance constraint. We show that pre-assigning these vertices
 to the two blocks of each bipartition (i.e., treating them as \emph{fixed vertices}) and
 optimizing the actual objective function on the remaining vertices
 yields provable balance guarantees for the resulting $k$-way partition.
 We implemented our algorithms in the open source HGP framework \Partitioner{KaHyPar}~\cite{KAHYPAR-DIS}.
 The experimental evaluation presented in Section~\ref{sec:experiments} shows
 that our new approach (called \Partitioner{KaHyPar-BP}) is able to compute balanced partitions
 for all instances of a large real-world benchmark set (without increasing the running time or decreasing the solution quality), while other
 partitioners such as the latest versions of \Partitioner{KaHyPar}, \Partitioner{hMetis}, and \Partitioner{PaToH} produced imbalanced partitions
 on $4.9\%$ up to $42\%$ of the instances for $\varepsilon = 0.01$ ($4.3\%$
 up to $23.1\%$ for $\varepsilon = 0.03$). Section~\ref{sec:conclusion} concludes the paper.

\section{Preliminaries}
\label{sec:preliminaries}

A \emph{weighted hypergraph} $H=(V,E,c,\omega)$ is defined as a set of vertices $V$ and a set of hyperedges/nets $E$ with vertex weights $c:V \to \mathbb{R}_{>0}$ and net weights $\omega:E \to \mathbb{R}_{>0}$, where each net $e$ is a subset of the vertex set $V$ (i.e., $e \subseteq V$).
We extend $c$ and $\omega$ to sets in the natural way, i.e., $c(U) :=\sum_{v\in U} c(v)$ and $\omega(F) :=\sum_{e \in F} \omega(e)$.
Given a subset $V' \subseteq V$, the \emph{subhypergraph}
$H_{V'}$ is defined as $H_{V'}:=(V', \{e \cap V' \mid e \in E : e \cap V' \neq \emptyset \}, c, \omega)$.

A \emph{$k$-way partition} of a hypergraph $H$ is a partition of the vertex set $V$ into $k$
non-empty disjoint subsets $\DefPartition{k} = \{V_1, \ldots, V_k\}$.
We refer to a $k$-way
partition $\DefPrepacking{k} = \{P_1, \ldots, P_k\}$ of a subset
$P \subseteq V$ as a \emph{$k$-way prepacking}. We call a vertex $v \in P$ a \emph{fixed} vertex
and a vertex $v \in V \setminus P$ an \emph{ordinary} vertex. During partitioning, fixed vertices are not allowed to be moved to a different block of the partition.
A $k$-way partition $\DefPartition{k}$ is \emph{$\varepsilon$-balanced} if each block $V_i$ satisfies the
\emph{balance constraint}: $c(V_i) \leq \balancedconstraint{k} := (1+\varepsilon)\lceil \frac{c(V)}{k} \rceil$ for
some parameter $\mathrm{\varepsilon}$.
The $k$\emph{-way hypergraph partitioning problem} initialized with a $k$-way prepacking $\DefPrepacking{k} = \{P_1, \ldots, P_k\}$ is to find
an $\varepsilon$-balanced $k$-way partition $\DefPartition{k} = \{V_1, \ldots, V_k\}$ of a hypergraph $H$ that minimizes an objective function and
satisfies that $\forall i \in \{1, \ldots, k\}: P_i \subseteq V_i$.
In this paper, we optimize the \emph{connectivity} metric $(\lambda - 1)(\Pi) := \sum_{e \in E} (\lambda(e) - 1) \: \omega(e)$,
where $\lambda(e) := |\{V_i \in \Pi \mid  V_i \cap e \neq \emptyset\}|$.

The \emph{most balanced partition problem} is to
find a $k$-way partition $\DefPartition{k}$ of a weighted hypergraph $H = (V,E,c,\omega)$ such that
$\max(\DefPartition{k}) := \max_{V' \in \DefPartition{k}} c(V')$ is minimized.
For an optimal solution
$\Partition_{\opt}$ it holds that there exists no other $k$-way partition $\DefPartition{k}'$ with
$\max(\DefPartition{k}') < \max(\Partition_{\opt})$. We use $\opt(H,k) := \max(\Partition_{\opt})$ to denote the weight of the heaviest block of an
optimal solution.
Note that the problem is equivalent to the
most common version of the \emph{job scheduling} problem: Given a sequence $J = \langle j_1, \ldots, j_n \rangle$ of $n$ computing jobs each
associated with a \emph{processing time} $p_i$ for $i \in [1,n]$, the task is to find an assignment of the $n$ jobs to
$k$ identical machines (each job $j_i$ runs exclusively on a machine for exactly $p_i$ time units) such that the latest completion time of a job is minimized.

\section{Related Work}
\label{sec:related_work}

In the following, we will focus on work closely related to our main contributions.
For an extensive overview on hypergraph partitioning we refer the reader to existing literature~\cite{DAlpert, DBLP:conf/dimacs/2012, DPapa2007, KAHYPAR-DIS}.
Well-known multilevel HGP software packages with certain distinguishing characteristics include
\Partitioner{PaToH}~\cite{DBLP:journals/jpdc/AykanatCU08, ccatalyurek1996decomposing} (originating from scientific computing),
\Partitioner{hMetis}~\cite{DBLP:conf/dac/KarypisAKS97, DBLP:journals/vlsi/KarypisK00} (originating from VLSI design),
\Partitioner{KaHyPar}~\cite{KaHyPar-MF-JEA,KaHyPar-CA} (general purpose, $n$-level),
\Partitioner{Moondrian}~\cite{DBLP:journals/siamrev/VastenhouwB05} (sparse matrix partitioning),
\Partitioner{UMPa}~\cite{DBLP:conf/dimacs/CatalyurekDKU12} (multi-objective) and
\Partitioner{Zoltan}~\cite{DBLP:conf/ipps/DevineBHBC06} (distributed partitioner).

\subparagraph*{Partitioning with Vertex Weights.}
The most widely used techniques to improve the quality of a $k$-way partition are move-based
local search heuristics~\cite{FiducciaM82,KLAlgorithm} that greedily move
vertices according to a \emph{gain} value (i.e., the improvement in the objective function).
Vertex moves violating the balance constraint are usually rejected,
which can significantly deteriorate solution
quality in presence of varying vertex weights~\cite{CaldwellKM00}. This issue is addressed using techniques
that allow intermediate balance violations~\cite{DBLP:conf/iccad/DuttT97} or
use temporary relaxations of the balance constraint~\cite{DBLP:conf/aspdac/CaldwellKM00, CaldwellKM00}.
Caldwell et al.~\cite{CaldwellKM00} proposed to
preassign each vertex with a weight greater than the average block weight $\balancedconstraint{k}$
to a seperate block before partitioning (treated as fixed vertices) and build the actual $k$-way partition around them.
All of these techniques were developed and evaluated for flat (i.e., non-multilevel)
partitioning algorithms. In the multilevel setting, even unweighted instances become
implicitly weighted due to vertex contractions in the coarsening phase, which is why
the formation of heavy vertices is prevented by penalizing the contraction of vertices with large weights~\cite{DBLP:journals/tpds/CatalyurekA99, Hauck:1996:MS:238604, DBLP:journals/tvlsi/ShinK93} or enforcing a strict upper bound for vertex weights throughout the coarsening process~\cite{KaHyPar-K, KaHyPar-CA}. If the input hypergraph is unweighted, the aforementioned techniques often suffice to find a
feasible solution~\cite{KAHYPAR-DIS}. \Partitioner{PaToH}~\cite{PaToHManual} additionally uses bin packing techniques during initial partitioning.

\subparagraph*{Job Scheduling Problem.}
The job scheduling problem is NP-hard~\cite{garey1979computers}
and we refer the reader to existing literature~\cite{graham1979optimization, pinedo2012scheduling} for a comprehensive overview of
the research topic. In this work, we make use of the \emph{longest processing time} (\lpt) algorithm
proposed by Graham~\cite{graham1969bounds}. We will explain the algorithm in the context of the most balanced
partition problem defined in Section~\ref{sec:preliminaries}: For a weighted hypergraph $H = (V,E,c,\omega)$,
the algorithm iterates over the vertices of $V$ sorted in decreasing vertex-weight order and assigns each vertex to the
block of the $k$-way partition with the lowest weight.
The algorithm can be implemented to run in $\Oh{|V|\log{|V|}}$ time, and for a $k$-way partition $\DefPartition{k}$ produced by the algorithm it holds that
$\max(\DefPartition{k}) \le (\frac{4}{3} - \frac{1}{3k})\opt(H,k)$.

\subparagraph*{KaHyPar.}
The \textbf{Ka}rlsruhe \textbf{Hy}pergraph \textbf{Par}titioning framework takes the multilevel paradigm to its extreme by only contracting a single vertex in every level of the
hierarchy. \Partitioner{KaHyPar} provides recursive bipartitioning~\cite{KaHyPar-R} as well as
direct $k$-way partitioning algorithms~\cite{KaHyPar-K} (direct $k$-way uses RB in the
initial partitioning phase). It uses a community detection algorithm as preprocessing step to restrict contractions to densely
connected regions of the hypergraph during coarsening~\cite{KaHyPar-CA}. Furthermore, it employs
a portfolio of bipartitioning algorithms for initial partitioning
of the coarsest hypergraph~\cite{HeuerBA, KaHyPar-R}, and,  during the refinement phase, improves the partition
with a highly engineered variant of the classical FM local search~\cite{KaHyPar-K} and a refinement technique
based on network flows~\cite{KAHYPAR-HFC,KaHyPar-MF-JEA}.

During RB-based partitioning, \Partitioner{KaHyPar} ensures that the solution is balanced by adapting the imbalance
ratio for each bipartition individually. Let $\subhypergraph{V'}$ be the subhypergraph of the current bipartition that should be partitioned
recursively into $k' \le k$ blocks. Then,
\begin{equation}
\label{eq:adaptive_imbalance}
\varepsilon' := \left( (1 + \varepsilon) \frac{c(V)}{k} \cdot \frac{k'}{c(V')} \right)^{\frac{1}{\lceil \log_2(k') \rceil}} - 1
\end{equation}
is the imbalance ratio used for the bipartition of $\subhypergraph{V'}$. The equation is based on the
observation that the worst-case block weight of the resulting $k'$-way partition of $\subhypergraph{V'}$
obtained via RB is smaller than $(1 + \varepsilon')^{\lceil \log_2(k') \rceil} \frac{c(V')}{k'}$,
if $\varepsilon'$ is used for all further bipartitions.
Requiring that this weight must be smaller or equal to $\balancedconstraint{k} = (1 + \varepsilon) \lceil \frac{c(V)}{k} \rceil$ leads
to the formula defined in Equation~\ref{eq:adaptive_imbalance}.

\section{A New Balance Constraint For Weighted Hypergraphs}
\label{sec:balance_constraint}

A $k$-way partition of a weighted hypergraph $H = (V,E,c,\omega)$ is balanced, if the weight of each block
is below some predefined upper bound. In the literature, the most commonly used bounds
are  $\balancedconstraint{k} := (1 + \varepsilon)\lceil \frac{c(V)}{k} \rceil$
(standard definition) and $\specialbalancedconstraint{k}{\max} := \balancedconstraint{k} + \max_{v \in V} c(v)$~\cite{FiducciaM82, KAHYPAR-DIS, Schulz2013_1000035713}.
The latter was initially proposed by Fiduccia and Mattheyses~\cite{FiducciaM82} for bipartitioning to
 ensure that the highest-gain vertex can always be moved to the opposite block.

Both definitions exhibit shortcomings in the presence of heavy vertices: As soon as the hypergraph contains even
a single vertex with $c(v) > \balancedconstraint{k}$, no feasible solution exists when the block weights are constrained
by $\balancedconstraint{k}$, while for $\specialbalancedconstraint{k}{\max}$ it follows that $\specialbalancedconstraint{k}{\max} > 2\balancedconstraint{k}$ -- allowing large variations in block weights even if $\varepsilon$ is small.
In the following, we therefore propose a new balance constraint that (i) guarantees the existence of
an $\varepsilon$-balanced $k$-way partition and (ii) avoids unnecessarily large imbalances.

While the optimal solution of the most balanced partition problem would yield a partition with the best possible balance,
it is not feasible in practice to use $\specialbalancedconstraint{k}{\opt} := (1 + \varepsilon)\opt(H,k)$ as balance constraint, because finding such a $k$-way partition is NP-hard~\cite{garey1979computers}. Hence, we
propose to use the bound provided by the $\lpt$ algorithm instead:
\begin{equation}
  \label{eq:new_balance_constraint}
\specialbalancedconstraint{k}{\lpt} := (1 + \varepsilon)~\lpt(H,k) \le \left( \frac{4}{3} - \frac{1}{3k} \right) \specialbalancedconstraint{k}{\opt}.
\end{equation}
Note that if the hypergraph is unweighted, the \lpt~algorithm will always
find an optimal solution with $\opt(H,k) = \lceil \frac{|V|}{k} \rceil$ and thus,
$\specialbalancedconstraint{k}{\lpt}$ is equal to $\balancedconstraint{k}$. Since all of today's partitioning algorithms
bound the maximum block weight by $\balancedconstraint{k}$, Section~\ref{sec:experiments} gives more details on how we employ this new balance constraint
definition in our experimental evaluation.

\section{Multilevel Recursive Bipartitioning with Vertex Weights Revisited}
\label{sec:balanced_partitioning}

\begin{figure}[t!]
  \centering
  \includegraphics[width=0.8\textwidth]{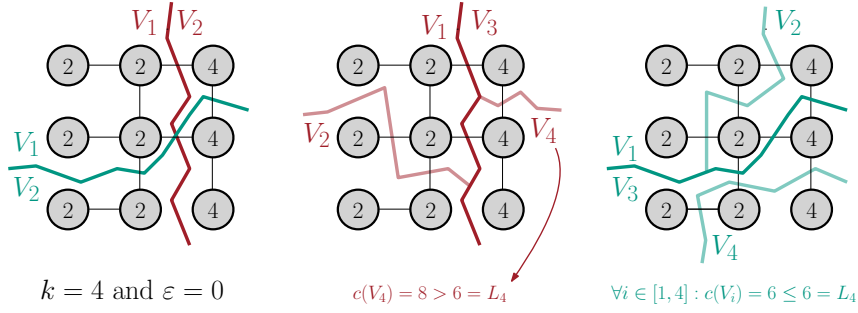}
  \caption{
    Illustration of a deeply (left, green line) and a non-deeply balanced bipartition (left, red line). The numbers
    in each circle denotes the vertex weights. In both cases, the
    hypergraph is partitioned into $k = 4$ blocks with $\varepsilon = 0$ via recursive bipartitioning.
    Thus, the weight of heaviest block must be smaller or equal to $\balancedconstraint{4} = 6$ and for the
    first bipartition, we use $\balancedconstraint{2} = 12$ as an upper bound.}
  \label{fig:deep_balance}
\end{figure}

Most multilevel hypergraph partitioners either employ recursive bipartitioning directly~\cite{ccatalyurek1996decomposing, DBLP:conf/ipps/DevineBHBC06, DBLP:conf/dac/KarypisAKS97, KaHyPar-R, DBLP:journals/siamrev/VastenhouwB05}
or use RB-based algorithms in the initial partitioning phase to compute an initial $k$-way partition of the coarsest
hypergraph~\cite{KaHyPar-K, DBLP:journals/jpdc/AykanatCU08, DBLP:conf/dimacs/CatalyurekDKU12, DBLP:journals/vlsi/KarypisK00}.
In both settings,  a $k$-way partition is derived by first computing a bipartition $\DefPartition{2} = \{V_1, V_2\}$ of the (input/coarse) hypergraph $H$
and then recursing on the subhypergraphs $\subhypergraph{V_1}$ and $\subhypergraph{V_2}$ by partitioning
$V_1$ into $\lceil \frac{k}{2} \rceil$ and $V_2$ into $\lfloor \frac{k}{2} \rfloor$ blocks. Although \Partitioner{KaHyPar} adaptively adjusts the allowed imbalance at each bipartitioning step (using the imbalance factor $\varepsilon'$ as defined in Equation~\ref{eq:adaptive_imbalance}),
an \emph{unfortunate} distribution of the vertices in some bipartitions \DefPartition{2} can easily lead to instances for
which it is impossible to find a balanced solution during the recursive calls -- even though the current bipartition \DefPartition{2} satisfies the adjusted balance constraint. An example is shown in Figure~\ref{fig:deep_balance} (left): Although the current bipartition (indicated by the red line) is perfectly balanced, it will not be possible to recursively partition the subhypergraph induced by the vertices of $V_2$ into two blocks of equal weight, because each of the three vertices has a weight of four.

To capture this problem, we introduce the notion of \textit{deep balance}:

\begin{definition}{\definitionname{Deep Balance}}
\label{def:deep_balance}
Let $H = (V,E,c,\omega)$ be a weighted hypergraph for which we want to compute an $\varepsilon$-balanced $k$-way partition,
and let $\subhypergraph{V'}$ be a subhypergraph of $H$ which should be partitioned into $k' \le k$ blocks via recursive bipartitioning.
A subhypergraph $\subhypergraph{V'}$ is deeply balanced w.r.t. $k'$, if there exists
a $k'$-way partition $\DefPartition{k'}$ of $\subhypergraph{V'}$ such that $\max(\DefPartition{k'}) \le
\balancedconstraint{k} := (1 + \varepsilon)\lceil \frac{c(V)}{k} \rceil$. A bipartition
$\DefPartition{2} = \{V_1, V_2\}$ of $\subhypergraph{V'}$ is deeply balanced w.r.t. $k'$, if the subhypergraphs
$\subhypergraph{V_1}$ and $\subhypergraph{V_2}$ are deeply balanced with respect to $\lceil \frac{k'}{2} \rceil$
resp.~$\lfloor \frac{k'}{2} \rfloor$.
\end{definition}

If a subhypergraph $\subhypergraph{V'}$ is deeply balanced with respect to $k'$, there always exists a $k'$-way
partition $\DefPartition{k'}$ of $\subhypergraph{V'}$ such that weight of the heaviest block satisfies the original balance
constraint \balancedconstraint{k} imposed on the partition of the input hypergraph $H$. Moreover, there also always exists
a deeply balanced bipartition $\DefPartition{2} := \{V_1,V_2\}$ ($V_1$ is the union of the first $\lceil \frac{k'}{2} \rceil$
and $V_2$ of the last $\lfloor \frac{k'}{2} \rfloor$ blocks of $\DefPartition{k'}$). Hence, a RB-based partitioning algorithm that
is able to compute deeply balanced bipartitions on deeply balanced subhypergraphs will always compute $\varepsilon$-balanced $k$-way partitions
(assuming the input hypergraph is deeply balanced).

\subparagraph*{Deep Balance and Adaptive Imbalance Adjustments.}
Computing deeply balanced bipartitions in the RB setting guarantees that the resulting $k$-way partition is $\varepsilon$-balanced.
Thus, the concept of deep balance could replace the adaptive imbalance factor $\varepsilon'$ employed in
\Partitioner{KaHyPar}~\cite{KaHyPar-R} (see Equation~\ref{eq:adaptive_imbalance}).
However, as we will see in the following example, combining both approaches gives the partitioner
more flexibility (in terms of feasible vertex moves during refinement).
Assume that we want to compute a $4$-way partition via recursive bipartitioning
and that the first bipartition $\DefPartition{2} := \{V_1, V_2\}$ is deeply balanced with
$c(V_1) = (1 + \varepsilon) \lceil \frac{c(V)}{2} \rceil$. The deep-balance property ensures that
we can further partition $V_1$ into two blocks such that the weight of the heavier block is smaller than $\balancedconstraint{4}$.
However, this bipartition has to be perfectly balanced:
\begin{equation} \balancedconstraint{2} = (1 + \overline{\varepsilon}) \Bigl\lceil \frac{c(V_1)}{2} \Bigr\rceil = (1 + \overline{\varepsilon}) \Bigl\lceil \frac{(1 + \varepsilon) \lceil \frac{c(V)}{2} \rceil}{2} \Bigr\rceil \le (1 + \varepsilon) \Bigl\lceil \frac{c(V)}{4} \Bigr\rceil = \balancedconstraint{4} \Rightarrow \overline{\varepsilon} \approx 0 .\end{equation}
If we would have computed the first bipartition with an adjusted imbalance factor $\varepsilon'$, then
$\max(\DefPartition{2}) \le (1 + \varepsilon') \lceil \frac{c(V)}{2} \rceil = \sqrt{1 + \varepsilon} \lceil \frac{c(V)}{2} \rceil$ -- providing more flexibility
for subsequent bipartitions. In the following, we therefore focus on computing deeply $\varepsilon'$-balanced bipartitions.

\subparagraph*{Deep Balance and Multilevel Recursive Bipartitioning.}
In general, computing a deeply balanced bipartition $\DefPartition{2} := \{V_1, V_2\}$ w.r.t.~$k$ is NP-hard, as we must show that there exists a $k$-way partition $\DefPartition{k}$
of $H$ with $\max(\DefPartition{k}) \le \balancedconstraint{k}$,
which can be reduced to the most balanced partition problem
presented in Section~\ref{sec:preliminaries}.
However, we can first compute a $k$-way partition
$\DefPartition{k} := \{V_1',\ldots,V_k'\}$ using the \lpt~algorithm,
thereby approximating an optimal solution.
If $\max(\DefPartition{k}) \le \balancedconstraint{k}$,
we can then construct a deeply balanced bipartition $\DefPartition{2} = \{V_1, V_2\}$ by choosing
$V_1 := V_1' \cup \ldots \cup V_{\lceil \frac{k}{2} \rceil}'$ and $V_2 := V_{\lceil \frac{k}{2} \rceil + 1}' \cup \ldots \cup V_k'$.
Unfortunately, this approach completely ignores the optimization of the objective function -- yielding balanced partitions of low quality.
If such a bipartition were to be used as initial solution in the multilevel setting, the objective could still be optimized during the
refinement phase. However, this would necessitate that refinement algorithms are aware of the concept of deep balance and that they only
perform vertex moves that don't destroy the deep-balance property of the starting solution. Since this is infeasible in practice,
we propose a different approach that involves fixed vertices.

The key idea of our approach is to compute a prepacking
$\Prepacking = \{P_1, P_2\}$ of the $m=|P_1|+|P_2|$ heaviest vertices
of the hypergraph and to show that this prepacking suffices to ensure that each $\varepsilon'$-balanced bipartition $\DefPartition{2} = \{V_1,V_2\}$
with $P_1 \subseteq V_1$ and $P_2 \subseteq V_2$ is deeply balanced.
Note that the upcoming definitions and theorems are formulated from the perspective of the first
bipartition of the input hypergraph $H$ to simplify notation. They can be generalized to subhypergraphs
$\subhypergraph{V'}$ in a similar fashion as was done in Definition~\ref{def:deep_balance}.
Furthermore, we say that the bipartition $\DefPartition{2} = \{V_1,V_2\}$ respects a prepacking $\Prepacking = \{P_1, P_2\}$,
if $P_1 \subseteq V_1$ and $P_2 \subseteq V_2$, and that the bipartition is balanced, if
$\max(\DefPartition{2}) \le \balancedconstraint{2} := (1 + \varepsilon') \lceil \frac{c(V_1 \cup V_2)}{2} \rceil$
(with $\varepsilon'$ as defined in Equation~\ref{eq:adaptive_imbalance}).
The following definition formalizes our idea.

\begin{definition}{\definitionname{Sufficiently Balanced Prepacking}}
Let $H = (V,E,c,\omega)$ be a hypergraph for which we want to compute an $\varepsilon$-balanced $k$-way partition via
recursive bipartitioning. We call a prepacking $\Prepacking$ of $H$ sufficiently
balanced if every balanced bipartition $\DefPartition{2}$ respecting $\Prepacking$ is deeply balanced
with respect to $k$.
\label{def:suffiently_balanced}
\end{definition}

Our approach to compute $\varepsilon$-balanced $k$-way partitions is outlined in Algorithm~\ref{algo:recursive_bipartitioning}.
We first compute a bipartition $\DefPartition{2}$. Before recursing on each of the two induced subhypergraphs, we check if $\DefPartition{2}$
is deeply balanced using the \lpt~algorithm in a similar fashion as described in the beginning of this paragraph. If it is not deeply balanced,
we compute a sufficiently balanced prepacking $\Prepacking$ and re-compute $\DefPartition{2}$ -- treating the vertices of the prepacking as fixed vertices. If this second bipartitioning call was able to compute a balanced bipartition, we found a deeply balanced partition and proceed to
partition the subhypergraphs recursively.

Note that, in general, we may not detect that $\DefPartition{2}$ is deeply balanced or
fail to find a sufficiently balanced prepacking $\Prepacking$ or a balanced bipartition
$\DefPartition{2}$, since all involved problems are NP-hard.
However, as we will see in Section~\ref{sec:experiments}, this only happens rarely in practice.

\begin{algorithm}
  \SetKwFunction{RB}{\small recursiveBipartitioning}
  \SetKwFunction{Bipartition}{\small multilevelBipartitioning}
  \SetKwFunction{Sufficient}{\small sufficientlyBalancedPrepacking}
  \SetKwProg{Fn}{Function}{:}{}
  \KwData{\small Hypergraph $H$ for which we seek an $\varepsilon$-balanced $k$-way
                 partition and subhypergraph $\subhypergraph{V'}$ of $H$ which is to be to bipartitioned
                 recursively into $k' \le k$ blocks.}
  \Fn{\RB{$H$, $k$, $\varepsilon$ $\subhypergraph{V'}$, $k'$}} {
    $\balancedconstraint{2} \gets (1 + \varepsilon') \lceil \frac{c(V')}{2} \rceil$ \tcp*{\footnotesize with $\varepsilon'$ as defined in Equation~\ref{eq:adaptive_imbalance}} \label{algo:balance_constraint}
    $\DefPartition{2} := \{V_1, V_2\} \gets \Bipartition{$\subhypergraph{V'}$, \balancedconstraint{2}, $\emptyset$}$ \tcp*{\footnotesize $\emptyset =$ empty prepacking} \label{algo:multilevel_bipartition}
    \lIf{ $k' = 2$ } {
      \Return{\DefPartition{2}}
    }
    \ElseIf{ $\DefPartition{2}$ is not deeply balanced w.r.t. $k'$ \label{algo:deep_balance_check} } { \label{algo:prepacking_triggered}
      $\Prepacking \gets \Sufficient{$H$, $k$, $\varepsilon$, $\subhypergraph{V'}$, $k'$}$ \tcp*{\footnotesize see Algorithm~\ref{algo:prepacking}} \label{algo:sufficiently_balanced_prepacking}
      $\DefPartition{2} \gets \Bipartition{$\subhypergraph{V'}$, \balancedconstraint{2}, $\Prepacking$}$ \tcp*{\footnotesize treating \Prepacking~as fixed vertices} \label{algo:restart_bipartition}
    }
    $\DefPartition{k_1} \gets \RB{$H$, $k$, $\varepsilon$, \subhypergraph{V_1}, $k_1$}$ with $k_1 := \lceil \frac{k'}{2} \rceil$ \;
    $\DefPartition{k_2} \gets \RB{$H$, $k$, $\varepsilon$, \subhypergraph{V_2}, $k_2$}$ with $k_2 := \lfloor \frac{k'}{2} \rfloor$ \;
    \Return{$\DefPartition{k_1} \cup \DefPartition{k_2}  $}
  }
  \caption{Recursive Bipartitioning Algorithm}
  \label{algo:recursive_bipartitioning}
 \end{algorithm}

\subparagraph*{Computing a Sufficiently Balanced Prepacking.}
The prepacking~\Prepacking~is constructed by incrementally assigning vertices to \Prepacking~in decreasing
order of weight and checking a property $\mathcal{P}$
after each assignment that, if satisfied,  implies that the current prepacking is sufficiently balanced. In the proof of property $\mathcal{P}$ , we will extend a $k$-way prepacking $\DefPrepacking{k}$ to an $\varepsilon$-balanced $k$-way partition
$\DefPartition{k}$ using the \lpt~algorithm and use the
following upper bound on the weight of the heaviest block of $\DefPartition{k}$.

\begin{restatable}{lemma}{worstcaseaf}{\definitionname{\lpt~Bound}}
\label{lemma:worstcaseaf}
Let $H = (V,E,c,\omega)$ be a weighted hypergraph, $\DefPrepacking{k}$ be a $k$-way prepacking for a set of
fixed vertices $\fixedvertices \subseteq V$, and let $\nonprepackedvertices:= \langle v_1, \ldots, v_m~|~v_i \in V \setminus \fixedvertices \rangle$ be the sequence of all ordinary
vertices of $V \setminus \fixedvertices$ sorted in decreasing order of weight.
If we assign the remaining vertices $\nonprepackedvertices$ to the blocks of $\DefPrepacking{k}$ by using the \lpt~algorithm,
we can extend $\DefPrepacking{k}$ to a $k$-way partition $\DefPartition{k}$ of $H$ such that the weight of the heaviest block
is bound by:
\vspace{-0.2cm}
\[\max(\DefPartition{k}) \le \max\{\frac{1}{k}c(\fixedvertices) + \afdbound{k}{\nonprepackedvertices}, \max(\DefPrepacking{k})\},
  \text{ with }
  \afdbound{k}{\nonprepackedvertices} := \max_{i \in \{1,\ldots,m\}} c(v_i) +
  \frac{1}{k} \sum_{j = 1}^{i - 1} c(v_j).\]
\end{restatable}

The proof of Lemma~\ref{lemma:worstcaseaf} can be found in Appendix~\ref{sec:af_proof}.
$\nonprepackedvertices$ is sorted in decreasing order of weight because for any permutation
$\nonprepackedvertices'$ of $\nonprepackedvertices$, it holds that $\afdbound{k}{\nonprepackedvertices} \le \afdbound{k}{\nonprepackedvertices'}$ --
resulting in the tightest bound for $\max(\DefPartition{k})$.

Assuming that the number $k$ of blocks is even (i.e., $k_1=k_2=\nicefrac{k}{2})$ to simplify notation, the balance property $\mathcal{P}$
is defined as follows (the generalized version can be found in Appendix~\ref{sec:generalized_balance_property}):

\begin{definition}{\definitionname{Balance Property $\mathcal{P}$}}
\label{def:balance_property}
Let $H = (V,E,c,\omega)$ be a hypergraph for which we want to compute an $\varepsilon$-balanced $k$-way
partition and let $\Prepacking$ be a prepacking of $H$ for a set of fixed vertices
$\fixedvertices \subseteq V$. Furthermore, let  $\heaviestvertices{t} := \langle v_1, \ldots, v_t \rangle$ be the sequence of
the $t$ heaviest ordinary vertices of $V \setminus \fixedvertices$ sorted in decreasing
order of weight such that $t$ is the smallest number that satisfies
$\max(\Prepacking) + c(\heaviestvertices{t}) \ge \balancedconstraint{2}$ (see Line~\ref{algo:balance_constraint}, Algorithm~\ref{algo:recursive_bipartitioning}).
We say that a prepacking $\Prepacking$ satisfies the balance property $\mathcal{P}$ if the following two conditions hold:
\begin{enumerate}
  \item[(i)] the prepacking~\Prepacking~is deeply balanced
  \item[(ii)] $\frac{1}{\nicefrac{k}{2}}\max(\Prepacking) + \afdbound{\nicefrac{k}{2}}{\heaviestvertices{t}} \le \balancedconstraint{k}$.
\end{enumerate}
\end{definition}

In the following, we will show that the \lpt~algorithm can be used to construct a $\nicefrac{k}{2}$-way partition $\DefPartition{\nicefrac{k}{2}}$ for both blocks of any
balanced bipartition $\DefPartition{2} = \{V_1, V_2\}$ that respects \Prepacking, such that the weight of the heaviest block can be bound by the left term of Condition~(ii).
This implies that $\max(\DefPartition{\nicefrac{k}{2}}) \le \balancedconstraint{k}$ (right term of Condition~(ii)) and thus
proofs that any balanced bipartition \DefPartition{2}~respecting \Prepacking~is deeply balanced.
Note that choosing $t$ as the smallest number that satisfies
$\max(\Prepacking) + c(\heaviestvertices{t}) \ge \balancedconstraint{2}$
minimizes the left term of Condition~(ii) (since $\afdbound{k}{\heaviestvertices{t}} \le \afdbound{k}{\heaviestvertices{t+1}}$).

\begin{theorem}
\label{theorem:sufficiently_balanced}
A prepacking $\Prepacking$ of a hypergraph $H = (V,E,c,\omega)$ that satisfies the balance property $\mathcal{P}$
is sufficiently balanced with respect to $k$.
\end{theorem}

\begin{proof}
For convenience, we use $k' := \nicefrac{k}{2}$. Let $\DefPartition{2} = \{V_1,V_2\}$ be an abitrary balanced
bipartition that respects the prepacking
$\Prepacking = \{P_1, P_2\}$ with $\max{(\DefPartition{2})} \le \balancedconstraint{2}$.
Since \Prepacking~is deeply balanced (see Definition~\ref{def:balance_property}(i)),
there exists a $k'$-way prepacking $\Prepacking_{k'}$ of $P_1$ such that $\max(\Prepacking_{k'}) \le \balancedconstraint{k}$.
We define the sequence of the ordinary
vertices of block $V_1$ sorted in decreasing weight order with
$\nonprepackedvertices_1 := \langle v_1, \ldots, v_m~|~v_i \in V_1 \setminus P_1 \rangle$.
We can extend $\Prepacking_{k'}$ to a $k'$-way partition $\DefPartition{k'}$ of $V_1$ by assigning the vertices of
$\nonprepackedvertices_1$ to the blocks in $\Prepacking_{k'}$ using the \lpt~algorithm. Lemma~\ref{lemma:worstcaseaf}
then establishes an upper bound on the weight of the heaviest block.
\[\max(\DefPartition{k'})
  \stackrel{\text{Lemma~\ref{lemma:worstcaseaf}}}{\le}
  \max\{\frac{1}{k'}c(P_1) + \afdbound{k'}{\nonprepackedvertices_1}, \max(\Prepacking_{k'})\}
  \stackrel{\scriptscriptstyle \max(\Prepacking_{k'}) \le \balancedconstraint{k}}{\le} \max\{\frac{1}{k'}c(P_1) +
  \afdbound{k'}{\nonprepackedvertices_1}, \balancedconstraint{k}\}\]

Let $\nonprepackedvertices_t$ be the sequence of the $t$ heaviest ordinary vertices of $V \setminus P$
with $P := P_1 \cup P_2$ as defined in Definition~\ref{def:balance_property}.

\begin{restatable}{claim}{balancepropertyclaim}
\label{claim:balancepropertyclaim}
It holds that: $\frac{1}{k'}c(P_1) + \afdbound{k'}{\nonprepackedvertices_1} \le \frac{1}{k'}\max(\Prepacking) + \afdbound{k'}{\heaviestvertices{t}}$.
\end{restatable}

For a proof of Claim~\ref{claim:balancepropertyclaim} see Appendix~\ref{sec:heaviest_elements_proof}.
We can conclude that
\[\frac{1}{k'}c(P_1) + \afdbound{k'}{\nonprepackedvertices_1} \stackrel{\text{Claim~\ref{claim:balancepropertyclaim}}}{\le}
  \frac{1}{k'}\max(\Prepacking) + \afdbound{k'}{\heaviestvertices{t}} \stackrel{\text{Definition~\ref{def:balance_property}(ii)}}{\le}
  \balancedconstraint{k}. \]
This proves that the subhypergraph $\subhypergraph{V_1}$ is deeply balanced. The proof for block $V_2$ can be done analogously,
which then implies that $\DefPartition{2}$ is deeply balanced. Since \DefPartition{2}~is an abitrary
balanced bipartition respecting \Prepacking, it follows that \Prepacking~is sufficiently balanced.
\end{proof}

Algorithm~\ref{algo:prepacking} outlines our approach to efficiently compute a sufficiently balanced prepacking \Prepacking.
In Line~\ref{algo:deep_balance_1}, we compute a $k'$-way prepacking $\DefPrepacking{k'}$ of the $i$ heaviest vertices with the \lpt~algorithm
and if $\DefPrepacking{k'}$ satisfies $\max(\DefPrepacking{k'}) \le \balancedconstraint{k}$, then
Line~\ref{algo:deep_balance_2} constructs a deeply balanced prepacking $\Prepacking$ (which fullfils Condition~(i) of Definition~\ref{def:balance_property}).
We store the blocks $P'_j$ of $\Prepacking_{k'}$ together with their weights $c(P'_j)$ as key in an addressable priority queue such that we can
determine and update the block with the smallest weight in time $\Oh{\log{k'}}$ (Line~\ref{algo:deep_balance_1}).
In Line~\ref{algo:min_t}, we compute the smallest $t$ that satisfies $\max(\Prepacking) + c(\heaviestvertices{t}) \ge \balancedconstraint{2}$
via a binary search in logarithmic time over an array containing the vertex weight prefix sums of the sequence $\nonprepackedvertices$, which
can be precomputed in linear time. Furthermore, we construct a range maximum query data structure
over the array $H_{\nicefrac{k'}{2}} = \langle c(v_1), c(v_2) + \frac{1}{\nicefrac{k'}{2}} c(v_1), \ldots,
c(v_n) + \frac{1}{\nicefrac{k'}{2}} \sum_{j = 1}^{n-1}c(v_j) \rangle$. Caculating $\afdbound{\nicefrac{k'}{2}}{\heaviestvertices{t}}$
(Line~\ref{algo:rmq_hk}) then corresponds to a range maximum query in the interval $[i + 1, i + t]$ in $H_{\nicefrac{k'}{2}}$, which can
be answered in constant time after $H_{\nicefrac{k'}{2}}$ has been precomputed in time $\Oh{n}$~\cite{bender2000lca}.
In total, the running time of the algorithm is $\Oh{n(\log{k'} + \log{n})}$.
Note that if the algorithm reaches Line~\ref{algo:failed}, we could not proof that any of the intermediate constructed
prepackings were sufficiently balanced, in which case $\Prepacking$ represents a bipartition
of $\subhypergraph{V'}$ computed by the \lpt~algorithm.

\begin{algorithm}
  \SetKwFunction{Sufficient}{\small sufficientlyBalancedPrepacking}
  \SetKwProg{Fn}{Function}{:}{}
  \KwData{\small Hypergraph $H = (V,E,c,\omega)$ for which we seek an $\varepsilon$-balanced $k$-way
                 partition and subhypergraph $\subhypergraph{V'} = (V',E',c,\omega)$ of $H$ which is to be to bipartitioned
                 recursively into $k' \le k$ blocks.}
  \Fn{\Sufficient{$H$, $k$, $\varepsilon$ $\subhypergraph{V'}$, $k'$}} {
    $\Prepacking = \langle P_1, P_2 \rangle \gets \langle \emptyset, \emptyset \rangle$ and $\Prepacking_{k'} = \langle P'_1, \ldots, P'_{k'} \rangle \gets \langle \emptyset, \ldots, \emptyset \rangle$ \tcp*{\footnotesize Initialization}
    $\balancedconstraint{2} \gets (1 + \varepsilon') \lceil \frac{c(V')}{2} \rceil$ and $\balancedconstraint{k} \gets (1 + \varepsilon) \lceil \frac{c(V)}{k} \rceil$ \tcp*{\footnotesize with $\varepsilon'$ as defined in Equation~\ref{eq:adaptive_imbalance}}
    $\nonprepackedvertices \gets \langle v_1, \ldots, v_n~|~v_i \in V'\rangle$ \tcp*{\footnotesize $V'$ sorted in decreasing order of weight $\Rightarrow \Oh{n\log{n}}$}

    \For{$i = 1$ to $n$}{
      Add $v_i \in \nonprepackedvertices$ to bin $P'_{j} \in \Prepacking_{k'}$ with smallest weight \label{algo:deep_balance_1} \tcp*{\footnotesize \lpt~algorithm}
      $\Prepacking \gets \{P_1' \cup \ldots \cup P_{x}', P_{x + 1}' \cup \ldots \cup P_{k'}' \}$ with $x:= \lceil\frac{k'}{2} \rceil$ \label{algo:deep_balance_2} \;
      \If(\tcp*[f]{\footnotesize $\Rightarrow \Prepacking$ is deeply ($\varepsilon'$-)balanced}){$\max(\Prepacking) \le \balancedconstraint{2}$ \text{and} $\max(\Prepacking_{k'}) \le \balancedconstraint{k}$}{
        $t \gets \min(\{t~|~\max(\Prepacking) + c(\heaviestvertices{t}) \ge \balancedconstraint{2}\})$ \label{algo:min_t} \tcp*{\footnotesize $\heaviestvertices{t} := \langle v_{i + 1}, \ldots, v_{i + t} \rangle$}
        \If(\tcp*[f]{\footnotesize Condition~(ii) of Definition~\ref{def:balance_property}}){$\frac{2}{k'}\max(\Prepacking) + \afdbound{\nicefrac{k'}{2}}{\heaviestvertices{t}} \le \balancedconstraint{k}$ \label{algo:rmq_hk}}{
          \Return{\Prepacking} \tcp*{\footnotesize $\Rightarrow$ \Prepacking~is sufficiently balanced (Theorem~\ref{theorem:sufficiently_balanced})}
        }
      }
    }
    \Return{\Prepacking} \tcp*{\footnotesize No sufficiently balanced prepacking found $\Rightarrow$ treat all vertices as fixed vertices} \label{algo:failed}
  }
  \caption{Prepacking Algorithm}
  \label{algo:prepacking}
 \end{algorithm}

\section{Experimental Evaluation}
\label{sec:experiments}

We integrated the prepacking technique (see Algorithms~\ref{algo:recursive_bipartitioning}~and~\ref{algo:prepacking}) into the recursive bipartitioning algorithm of \Partitioner{KaHyPar}.
Our implementation is available from \url{http://www.kahypar.org}. The code is written in \Cpp{17} and compiled using \gpp{9.2} with the flags \texttt{-mtune=native -O3  -march=native}. Since \Partitioner{KaHyPar} offers both a recursive bipartitioning and direct $k$-way partitioning algorithm (which uses the RB algorithm in the initial partitioning phase), we refer to the
RB-version using our improvements as \Partitioner{KaHyPar-BP-R} and to the direct $k$-way version as \Partitioner{KaHyPar-BP-K} (\texttt{BP} $=$ \textbf{B}alanced \textbf{P}artitioning).

\subparagraph*{Instances.}
The following experimental evaluation is based on two benchmark sets.
The \InstanceType{RealWorld} benchmark set consists of 50 hypergraphs originating from the VLSI design and
scientific computing domain. It contains instances from the ISPD98 VLSI Circuit Benchmark Suite~\cite{ISPD98}
(18 instances),
the DAC 2012 Routability-Driven Placement Benchmark Suite~\cite{DAC2012} (9 instances),
16 instances from the Stanford Network Analysis (SNAP) Platform~\cite{SNAP}, and
7 highly asymmetric matrices of Davis et al.~\cite{AsymmetricSPM} (referred to as \InstanceType{ASM}).
For VLSI instances (\InstanceType{ISPD98} and \InstanceType{DAC}), we use the area
of a circuit element as the weight of its corresponding vertex.
We translate sparse matrices (\InstanceType{SNAP} and \InstanceType{ASM} instances) to hypergraphs using the row-net model~\cite{DBLP:journals/tpds/CatalyurekA99}
and use the degree of a vertex as its weight.
The vertex weight distributions of the individual instance types are depicted in Figure~\ref{fig:vertex_weights} in Appendix~\ref{sec:vertex_weight_distribution}.
\footnote{The benchmark sets and detailed statistics of their properties
are publicly available from \url{http://algo2.iti.kit.edu/heuer/sea21/}.}

Additionally, we generate ten \InstanceType{Artificial} instances that use the net structure of the ten largest \InstanceType{ISPD98} instances. Instead of using the area as weight, we assign new vertex weights that
yield instances for which it is difficult to satisfy the balance constraint:
Each vertex is assigned either unit weight or a weight chosen randomly from an uniform distribution in $[1, W] \subseteq \mathbb{N}_+$.
Both the probability that a vertex has non-unit weight and the parameter $W$ are determined (depending on the total number of vertices) such that
the expected number of vertices with non-unit weight is 120 and the expected total weight of these vertices is half the expected total weight of
the resulting hypergraph.

\subparagraph*{System and Methodology.}
All experiments are performed on a single core of a cluster with Intel Xeon Gold 6230 processors
running at $2.1$ GHz with $96$GB RAM. We compare \Partitioner{KaHyPar-BP-R} and \Partitioner{KaHyPar-BP-K} with
the latest recursive bipartitioning ($\Partitioner{KaHyPar-R}$) and direct $k$-way version (\Partitioner{KaHyPar-K}) of \Partitioner{KaHyPar}~\cite{KAHYPAR-HFC},
the default (\Partitioner{PaToH-D}) and quality preset (\Partitioner{PaToH-Q}) of \Partitioner{PaToH} 3.3~\cite{ccatalyurek1996decomposing},
as well as with the recursive bipartitioning (\Partitioner{hMetis-R}) and direct $k$-way version (\Partitioner{hMetis-K}) of
\Partitioner{hMetis} 2.0~\cite{DBLP:conf/dac/KarypisAKS97, DBLP:journals/vlsi/KarypisK00}. Details about the
choices of config parameters that influence partitioning quality or imbalance can be found in Appendix~\ref{sec:evaluated_partitioner_config}.

We perform experiments using $\splitatcommas{k \in \{2,4,8,16,32,64,128\}}$, $\splitatcommas{\varepsilon \in \{0.01, 0.03, 0.1\}}$,
ten repetitions using different seeds for each combination of $k$ and $\varepsilon$, and a time limit of eight hours. We call a combination of a hypergraph $H = (V,E,c,\omega)$,
$k$, and $\varepsilon$ an \emph{instance}. Before partitioning an instance, we remove all vertices $v \in V$ from $H$
with a weight greater than $\balancedconstraint{k} = (1 + \varepsilon)\lceil \frac{c(V)}{k} \rceil$ as proposed by
Caldwell et al.~\cite{CaldwellKM00} and adapt $k$ to $k' := k - |V_R|$, where $V_R$ represents
the set of removed vertices. We repeat that step recursively until there is no vertex with a weight greater than
$\balancedconstraint{k'} := (1 + \varepsilon)\lceil \frac{c(V \setminus V_R)}{k'} \rceil$. The input for each partitioner
is the subhypergraph $\subhypergraph{V \setminus V_R}$ of $H$ for which we compute a $k'$-way partition with
\specialbalancedconstraint{k'}{\lpt} as maximum allowed block weight.
Note that since all evaluated partitioners internally employ $\balancedconstraint{k'}$ as balance constraint,
we initialize each partitioner with a modified imbalance factor $\hat{\varepsilon}$ instead of $\varepsilon$ which is calculated as follows:
\begin{equation*}
\label{eq:modified_imbalance}
\balancedconstraint{k'} = (1 + \hat{\varepsilon}) \left\lceil \frac{c(V \setminus V_R)}{k'} \right\rceil = (1 + \varepsilon)\lpt(\subhypergraph{V \setminus V_R},k') = \specialbalancedconstraint{k'}{\lpt} \Rightarrow
\hat{\varepsilon} = \frac{\specialbalancedconstraint{k'}{\lpt}}{\lceil \frac{c(V \setminus V_R)}{k'} \rceil} - 1.
\end{equation*}
We consider the resulting $k'$-way partition $\DefPartition{k'}$ to be imbalanced, if it is not $\hat{\varepsilon}$-balanced.
Each partitioner optimizes the connectivity metric, which we also refer to as the quality of a partition.
Partition $\DefPartition{k'}$ can be extended to a $k$-way partition $\DefPartition{k}$ by
adding each of the removed vertices $v \in V_R$ to $\DefPartition{k}$ as a separate block.
Note that adding the removed vertices increases the connectivity metric of a $k'$-way
partition only by a constant value $\alpha \ge 0$.
Thus, we report the quality of $\DefPartition{k'}$, since $(\lambda - 1)(\DefPartition{k})$ will be always equal
to $(\lambda - 1)(\DefPartition{k'}) + \alpha$.

For each instance, we average quality and running times using the arithmetic
mean (over all seeds).
To further average over multiple instances, we use the geometric mean for absolute running times
to give each instance a comparable influence.
Runs with imbalanced partitions are not excluded from averaged running times.
If \textit{all ten runs} of a partitioner produced imbalanced partitions on an instance, we consider
the instance as \emph{imbalanced} and mark it with \ding{55} in the plots.

To compare the solution quality of different algorithms, we use \emph{performance profiles}~\cite{DBLP:journals/mp/DolanM02}.
Let $\mathcal{A}$ be the set of all algorithms we want to compare, $\mathcal{I}$ the set of instances, and $q_{A}(I)$ the quality of algorithm
$A \in \mathcal{A}$ on instance $I \in \mathcal{I}$.
For each algorithm $A$, we plot the fraction of instances ($y$-axis) for which $q_A(I) \leq \tau \cdot \min_{A' \in \mathcal{A}}q_{A'}(I)$, where $\tau$ is on the $x$-axis.
For $\tau = 1$, the $y$-value indicates the percentage of instances for which an algorithm $A \in \mathcal{A}$ performs best.
Note that these plots relate the quality of an algorithm to the best solution and thus do not permit a full ranking of three or more algorithms.

\subparagraph*{Balanced Partitioning.}
In Table~\ref{table:imbalanced_instances_per_type}, we report the percentage of imbalanced instances
produced by each partitioner for each instance type and $\varepsilon$. Both \Partitioner{KaHyPar-BP-K} and
\Partitioner{KaHyPar-BP-R} compute balanced partitions for all tested benchmark sets and parameters.  For the remaining partitioners, the number of
imbalanced solutions increases as the balance constraint becomes tighter.
For the previous KaHyPar versions, the number of imbalanced partitions is most pronounced on VLSI instances: For $\varepsilon = 0.01$,
\Partitioner{KaHyPar-K} and \Partitioner{KaHyPar-R} compute infeasible solutions for 6.3\% (10.3\%) of the \InstanceType{ISPD98} and for 9.5\% (19.0\%) of the \InstanceType{DAC} instances.
Comparing the distribution of vertex weights reveals that these instances tend to have a larger proportion of \emph{heavier} vertices
compared to the \InstanceType{ASM} and \InstanceType{SNAP} instances (see Figure~\ref{fig:vertex_weights} in Appendix~\ref{sec:vertex_weight_distribution}).
The largest benefit of using our approach can be observed on the artificially generated instances, where
\Partitioner{KaHyPar-K} and \Partitioner{KaHyPar-R} only computed balanced partitions for 72.9\% (71.4\%) of
the instances (for $\varepsilon = 0.01$).

\begin{table}[!t]
  \centering
  \caption{Percentage of instances for which all ten computed partitions were imbalanced.}
  \label{table:imbalanced_instances_per_type}
  \scalebox{0.82}{
  \begin{tabular}{r|rrr|rrr|rrr|rrr|rrr}
                                            & \multicolumn{3}{c|}{\InstanceType{ISPD98}} &
                                              \multicolumn{3}{c|}{\InstanceType{DAC}} &
                                              \multicolumn{3}{c|}{\InstanceType{ASM}} &
                                              \multicolumn{3}{c|}{\InstanceType{SNAP}} &
                                              \multicolumn{3}{c}{\InstanceType{Artificial}} \\
  $\varepsilon$                             & \scriptsize{0.01} & \scriptsize{0.03}  & \scriptsize{0.1} &
                                              \scriptsize{0.01} & \scriptsize{0.03}  & \scriptsize{0.1} &
                                              \scriptsize{0.01} & \scriptsize{0.03}  & \scriptsize{0.1} &
                                              \scriptsize{0.01} & \scriptsize{0.03}  & \scriptsize{0.1} &
                                              \scriptsize{0.01} & \scriptsize{0.03}  & \scriptsize{0.1} \\
  \midrule
  	% Columns: ispd, dac, asm, snap, artificial 
 \Partitioner{\Partitioner{KaHyPar-BP-K}} & 0.0 & 0.0 & 0.0 & 0.0 & 0.0 & 0.0 & 0.0 & 0.0 & 0.0 & 0.0 & 0.0 & 0.0 & 0.0 & 0.0 & 0.0 \tabularnewline
 \Partitioner{\Partitioner{KaHyPar-BP-R}} & 0.0 & 0.0 & 0.0 & 0.0 & 0.0 & 0.0 & 0.0 & 0.0 & 0.0 & 0.0 & 0.0 & 0.0 & 0.0 & 0.0 & 0.0 \tabularnewline
 \Partitioner{\Partitioner{KaHyPar-K}} & 6.3 & 5.6 & 0.8 & 9.5 & 7.9 & 6.3 & 4.1 & 4.1 & 2.0 & 0.9 & 0.9 & 0.0 & 27.1 & 22.9 & 11.4 \tabularnewline
 \Partitioner{\Partitioner{KaHyPar-R}} & 10.3 & 8.7 & 7.1 & 19.0 & 19.0 & 14.3 & 6.1 & 4.1 & 4.1 & 6.2 & 2.7 & 0.9 & 28.6 & 24.3 & 12.9 \tabularnewline
 \Partitioner{\Partitioner{hMetis-K}} & 43.7 & 22.2 & 9.5 & 33.3 & 22.2 & 11.1 & 67.3 & 32.7 & 4.1 & 33.9 & 20.5 & 3.6 & 51.4 & 38.6 & 24.3 \tabularnewline
 \Partitioner{\Partitioner{hMetis-R}} & 17.5 & 15.1 & 7.1 & 20.6 & 15.9 & 12.7 & 8.2 & 6.1 & 4.1 & 15.2 & 10.7 & 4.5 & 58.6 & 54.3 & 34.3 \tabularnewline
 \Partitioner{\Partitioner{PaToH-Q}} & 15.9 & 11.1 & 5.6 & 23.8 & 17.5 & 9.5 & 24.5 & 6.1 & 4.1 & 33.9 & 8.0 & 1.8 & 31.4 & 24.3 & 14.3 \tabularnewline
 \Partitioner{\Partitioner{PaToH-D}} & 9.5 & 7.9 & 3.2 & 20.6 & 17.5 & 9.5 & 28.6 & 6.1 & 4.1 & 22.3 & 11.6 & 2.7 & 20.0 & 15.7 & 8.6 \tabularnewline

  \end{tabular}
  }
\end{table}

\begin{table}[!t]
  \centering
  \caption{
  Occurrence of prepacked vertices (i.e., vertices that are fixed to a specific block during partitioning) for each combination of $k$ and $\varepsilon$ when using \Partitioner{KaHyPar-BP-R} on \InstanceType{RealWorld} instances:
  Minimum/average/maximum percentage of prepacked vertices (left), and percentage of instances for which the prepacking is executed at least once (right).
  }
  \label{table:prepacking_stats}
  \scalebox{0.95}{
  \begin{tabular}{r|rrrrrrrrr||rrr}
       &  \multicolumn{3}{c}{$\varepsilon = 0.01$} & \multicolumn{3}{c}{$\varepsilon = 0.03$} & \multicolumn{3}{c||}{$\varepsilon = 0.1$} & \multicolumn{3}{c}{Prepacking Triggered} \\
  $k$  &   \scriptsize{Min} & \scriptsize{Avg} & \scriptsize{Max} &
          \scriptsize{Min} & \scriptsize{Avg} & \scriptsize{Max} &
          \scriptsize{Min} & \scriptsize{Avg} & \scriptsize{Max} &
          \scriptsize{$\varepsilon = 0.01$} & \scriptsize{$\varepsilon = 0.03$} & \scriptsize{$\varepsilon = 0.1$} \\
  \midrule
  	2 & - & - & - & - & - & - & - & - & - & - & - & - \tabularnewline
4 & - & - & - & - & - & - & - & - & - & - & - & - \tabularnewline
8 & $\scriptscriptstyle \le~0.1$ & $\scriptscriptstyle \le~0.1$ & 0.2 & $\scriptscriptstyle \le~0.1$ & $\scriptscriptstyle \le~0.1$ & $\scriptscriptstyle \le~0.1$ & - & - & - & 5.0 & 1.7 & - \tabularnewline
16 & $\scriptscriptstyle \le~0.1$ & $\scriptscriptstyle \le~0.1$ & $\scriptscriptstyle \le~0.1$ & $\scriptscriptstyle \le~0.1$ & $\scriptscriptstyle \le~0.1$ & $\scriptscriptstyle \le~0.1$ & $\scriptscriptstyle \le~0.1$ & $\scriptscriptstyle \le~0.1$ & $\scriptscriptstyle \le~0.1$ & 8.3 & 5.0 & 3.3 \tabularnewline
32 & $\scriptscriptstyle \le~0.1$ & 8.1 & 59.0 & $\scriptscriptstyle \le~0.1$ & 6.2 & 68.4 & $\scriptscriptstyle \le~0.1$ & 1.9 & 14.7 & 20.0 & 18.3 & 10.0 \tabularnewline
64 & $\scriptscriptstyle \le~0.1$ & 23.2 & 87.7 & $\scriptscriptstyle \le~0.1$ & 17.3 & 90.9 & $\scriptscriptstyle \le~0.1$ & 2.7 & 35.9 & 18.3 & 13.3 & 10.0 \tabularnewline
128 & $\scriptscriptstyle \le~0.1$ & 67.9 & 100.0 & $\scriptscriptstyle \le~0.1$ & 42.0 & 96.3 & $\scriptscriptstyle \le~0.1$ & 15.4 & 97.0 & 26.7 & 20.0 & 15.0 \tabularnewline

  \end{tabular}
  }
\end{table}

With some notable exceptions, the number of imbalanced partitions of both variants of \Partitioner{PaToH} and \Partitioner{hMetis-R} is comparable to that of \Partitioner{KaHyPar-R}: \Partitioner{PaToH} computes significantly fewer feasible solutions on sparse matrix
instances (\InstanceType{ASM} and \InstanceType{SNAP}) for $\varepsilon = 0.01$, while \Partitioner{hMetis-R} performs considerably worse on the \InstanceType{Artificial} benchmark set. Out of all partitioners, \Partitioner{hMetis-K} yields the most imbalanced instances across all benchmark sets.
As can be seen in Table~\ref{table:partitioner_stats} in Appendix~\ref{sec:imbalance_per_type}, the number of imbalanced partitions produced by each competing partitioner increases with deceasing $\varepsilon$ and
increasing $k$.

Table~\ref{table:prepacking_stats} shows (i) how often our prepacking algorithm is
triggered at least once in \Partitioner{KaHyPar-BP-R} (see Line~\ref{algo:prepacking_triggered} in Algorithm~\ref{algo:recursive_bipartitioning})
and (ii) the percentage of vertices that are treated as fixed vertices
(see Table~\ref{table:prepacking_stats_kahypar_bp_k} in Appendix~\ref{sec:prepacking_stats_kahypar_bp_k} for the results of \Partitioner{KaHyPar-BP-K}).
Except for $k = 128$, on average less than 25\% of the vertices are treated as fixed vertices (even less than 10\% for $k < 64$),
which provides sufficient flexibility to optimize the connectivity objective on the remaining ordinary vertices.
However, in a few cases there are also runs where almost all vertices are added to the prepacking.
As expected, the triggering frequency and the percentage of fixed vertices increases for larger values of $k$ and smaller $\varepsilon$.

\subparagraph*{Quality and Running Times.}
Comparing the different KaHyPar configurations in Figure~\ref{fig:quality} (left), we can see that our new configurations
provide the same solution quality as their non-prepacking counterparts. Furthermore, we see that, in general,
the direct $k$-way algorithm still performs better than its RB counterpart~\cite{KAHYPAR-DIS}.
Figure~\ref{fig:quality} (middle) therefore compares the strongest configuration \Partitioner{KaHyPar-BP-K} with \Partitioner{PaToH} and \Partitioner{hMetis}. We see that  \Partitioner{KaHyPar-BP-K} performs considerably better than the competitors.
If we compare \Partitioner{KaHyPar-BP-K} with each partitioner individually on the \InstanceType{RealWorld} benchmark set,
\Partitioner{KaHyPar-BP-K} produces partitions with higher quality than those of \Partitioner{KaHyPar-K}, \Partitioner{KaHyPar-BP-R},
\Partitioner{KaHyPar-R}, \Partitioner{hMetis-R}, \Partitioner{hMetis-K}, \Partitioner{PaToH-Q}
and \Partitioner{PaToH-D} on 48.9\%, 70.2\%, 73.2\%, 76.4\%, 84.3\%, 92.9\% and 97.9\% of the instances, respectively.
\Partitioner{KaHyPar-BP-K} outperforms \Partitioner{KaHyPar-BP-R} on the \InstanceType{RealWorld} benchmark set. On artificial instances,
both algorithms produce partitions with comparable quality for $\varepsilon = \{0.01, 0.03\}$, while the results
are less clear for $\varepsilon = 0.1$ (see Figure~\ref{fig:quality} (right), as well as
Figures~\ref{fig:quality_epsilon_real_world}~and~\ref{fig:quality_epsilon_artificial} in Appendix~\ref{sec:quality_comparison}).

The running time plots (see Figure~\ref{fig:running_time_real_world}~and~\ref{fig:running_time_artificial} in Appendix~\ref{sec:running_time}) show that our new approach
does not impose any additional overheads in \Partitioner{KaHyPar}. On average, \Partitioner{KaHyPar-BP-K} is slightly faster than
\Partitioner{KaHyPar-K} as our new algorithm has replaced the previous balancing strategy in \Partitioner{KaHyPar}
(restarting the bipartition with an tighter bound on the weight of the heaviest block if the bipartition is imbalanced).
The running time difference is less pronounced for \Partitioner{KaHyPar-BP-R} and \Partitioner{KaHyPar-R}. This can be explained
by the fact that, in  \Partitioner{KaHyPar-BP-R}, our prepacking algorithm is executed on the input hypergraph, whereas it is executed
on the coarsest hypergraph in \Partitioner{KaHyPar-BP-K}.

\begin{figure*}[!t]
  \hspace{-0.25cm}
  \begin{minipage}{.365\textwidth}
  	\externalizedfigure{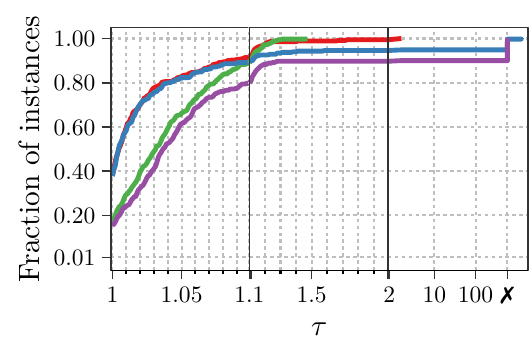}{plots/realworld_e0.01_kahypar}
  \end{minipage} %
  \begin{minipage}{.305\textwidth}
  	\externalizedfigure{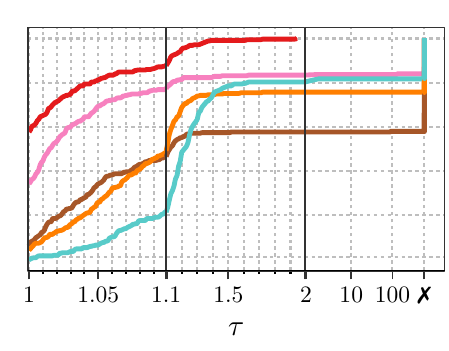}{plots/realworld_e0.01_comparison}
  \end{minipage} %
  \begin{minipage}{.305\textwidth}
    \externalizedfigure{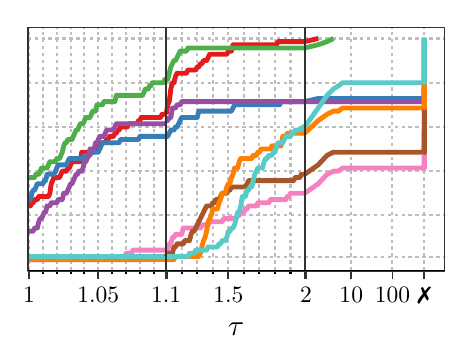}{plots/artificial_e0.01}
  \end{minipage} %
	\vfill
	\vspace{-0.4cm}
  \begin{minipage}{\textwidth}
    \externalizedfigure{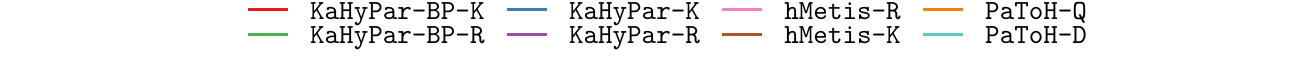}{plots/legend}
  \end{minipage} %
  \vspace{-0.25cm}
  \caption{Performance profiles comparing the solution quality of \Partitioner{KaHyPar-BP-K} and
           \Partitioner{KaHyPar-BP-R} with \Partitioner{KaHyPar-K} (left), \Partitioner{KaHyPar-R} (left),
           \Partitioner{PaToH} (middle), and \Partitioner{hMetis} (middle) on our \InstanceType{RealWorld}
           benchmark set, and with all systems on our \InstanceType{Artificial} benchmark set (right)
          ($\varepsilon = 0.01$).}
  \label{fig:quality}
  \vspace{-0.5cm}
\end{figure*}

\section{Conclusion and Future Work}
\label{sec:conclusion}

In this work, we revisited the problem of computing balanced partitions for weighted hypergraphs in the
multilevel setting and showed that many state-of-the-art hypergraph partitioners struggle to find balanced solutions
on hypergraphs with weighted vertices -- especially for tight balance constraints.
We therefore developed an algorithm that enables partitioners based on the recursive bipartitioning
scheme to reliably compute balanced partitions. The method is based on the concept of
\emph{deeply balanced} bipartitions and is implemented by pre-assigning a small subset of the heaviest vertices
to the two blocks of each bipartiton. For this pre-assignment, we established a property that
can be verified in polynomial time and, if fulfilled, leads to provable balance guarantees for the resulting
$k$-way partition. We integrated the approach into the recursive bipartitioning algorithm of \Partitioner{KaHyPar}.
Our new algorithms  \Partitioner{KaHyPar-BP-K} and  \Partitioner{KaHyPar-BP-R} are capable of computing balanced solutions on all instances of a diverse benchmark set, without negatively affecting the solution quality or running time
of \Partitioner{KaHyPar}.

Interesting opportunities for future research include replacing the \lpt~algorithm with
an algorithm that additionally optimizes the partitioning objective to construct sufficiently balanced prepackings
with improved solution quality~\cite{MaasBA}, and integrating rebalancing strategies
similar to the techniques proposed for non-multilevel
partitioners~\cite{DBLP:conf/aspdac/CaldwellKM00, CaldwellKM00, DBLP:conf/iccad/DuttT97} into multilevel refinement algorithms.

\bibliography{multilevel_vertex_weights}

\begin{thebibliography}{10}

\bibitem{KaHyPar-K}
Y.~Akhremtsev, T.~Heuer, P.~Sanders, and S.~Schlag.
\newblock {Engineering a Direct \emph{k}-way Hypergraph Partitioning
  Algorithm}.
\newblock In {\em 19th Workshop on Algorithm Engineering and Experiments
  (ALENEX)}, pages 28--42. SIAM, 01 2017.

\bibitem{ISPD98}
C.~J. Alpert.
\newblock {The ISPD98 Circuit Benchmark Suite}.
\newblock In {\em International Symposium on Physical Design (ISPD)}, pages
  80--85, 4 1998.

\bibitem{DAlpert}
C.~J. Alpert and A.~B. Kahng.
\newblock {Recent Directions in Netlist Partitioning: A Survey}.
\newblock {\em Integration: The VLSI Journal}, 19(1-2):1--81, 1995.

\bibitem{DBLP:journals/jpdc/AykanatCU08}
C.~Aykanat, B.~B. Cambazoglu, and B.~U{\c{c}}ar.
\newblock {Multi-Level Direct $k$-Way Hypergraph Partitioning with Multiple
  Constraints and Fixed Vertices}.
\newblock {\em Journal of Parallel and Distributed Computing}, 68(5):609--625,
  2008.

\bibitem{DBLP:conf/dimacs/2012}
D.~A. Bader, H.~Meyerhenke, P.~Sanders, and D.~Wagner, editors.
\newblock {\em Graph Partitioning and Graph Clustering, 10th {DIMACS}
  Implementation Challenge Workshop}, volume 588 of {\em Contemporary
  Mathematics}. American Mathematical Society, 2 2013.

\bibitem{bender2000lca}
M.~A. Bender and M.~Farach-Colton.
\newblock {The LCA Problem Revisited}.
\newblock In {\em Latin American Symposium on Theoretical Informatics}, pages
  88--94. Springer, 2000.

\bibitem{bisseling2012two}
R.~H. Bisseling, B.~O.~Auer Fagginger, A.~N. Yzelman, T.~van Leeuwen, and
  {\"U}.~V. {\c{C}}ataly{\"u}rek.
\newblock {Two-Dimensional Approaches to Sparse Matrix Partitioning}.
\newblock {\em Combinatorial Scientific Computing}, pages 321--349, 2012.

\bibitem{DBLP:journals/ipl/BuiJ92}
T.~N. Bui and C.~Jones.
\newblock {Finding Good Approximate Vertex and Edge Partitions is NP-Hard}.
\newblock {\em Information Processing Letters}, 42(3):153--159, 05 1992.

\bibitem{DBLP:conf/aspdac/CaldwellKM00}
A.~E. Caldwell, A.~B. Kahng, and I.~L. Markov.
\newblock {Improved Algorithms for Hypergraph Bipartitioning}.
\newblock In {\em Asia South Pacific Design Automation Conference (ASP-DAC)},
  pages 661--666, 2000.

\bibitem{CaldwellKM00}
A.~E. Caldwell, A.~B. Kahng, and I.~L. Markov.
\newblock {Iterative Partitioning with Varying Node Weights}.
\newblock {\em {VLSI} Design}, (3):249--258, 2000.

\bibitem{ccatalyurek1996decomposing}
{\"U}.~V {\c{C}}ataly{\"u}rek and C.~Aykanat.
\newblock {Decomposing Irregularly Sparse Matrices for Parallel Matrix-Vector
  Multiplication}.
\newblock In {\em International Workshop on Parallel Algorithms for Irregularly
  Structured Problems}, pages 75--86. Springer, 1996.

\bibitem{PaToHManual}
{\"{U}}.~V. {\c{C}}ataly{\"{u}}rek and C.~Aykanat.
\newblock {PaToH: Partitioning Tool for Hypergraphs}.
\newblock \url{https://www.cc.gatech.edu/~umit/PaToH/manual.pdf}, 2011.

\bibitem{DBLP:journals/tpds/CatalyurekA99}
{\"{U}}.~V. {\c{C}}ataly{\"{u}}rek and Cevdet Aykanat.
\newblock {Hypergraph-Partitioning-Based Decomposition for Parallel
  Sparse-Matrix Vector Multiplication}.
\newblock {\em IEEE Transactions on Parallel and Distributed Systems},
  10(7):673--693, 1999.

\bibitem{DBLP:conf/dimacs/CatalyurekDKU12}
{\"{U}}.~V. {\c{C}}ataly{\"{u}}rek, M.~Deveci, K.~Kaya, and B.~U{\c{c}}ar.
\newblock {UMPa: A Multi-Objective, Multi-Level Partitioner for Communication
  Minimization}.
\newblock In {\em Graph Partitioning and Graph Clustering, 10th {DIMACS}
  Implementation Challenge Workshop}, pages 53--66, 2 2012.

\bibitem{AsymmetricSPM}
T.~Davis, I.~S. Duff, and S.~Nakov.
\newblock {Design and Implementation of a Parallel Markowitz Threshold
  Algorithm}.
\newblock {\em SIAM Journal on Matrix Analysis and Applications},
  41(2):573--590, 4 2020.

\bibitem{DBLP:conf/ipps/DevineBHBC06}
K.~D. Devine, E.~G. Boman, R.~T. Heaphy, R.~H. Bisseling, and {\"{U}}.~V.
  {\c{C}}ataly{\"{u}}rek.
\newblock {Parallel Hypergraph Partitioning for Scientific Computing}.
\newblock In {\em 20th International Parallel and Distributed Processing
  Symposium (IPDPS)}, 4 2006.

\bibitem{DBLP:journals/mp/DolanM02}
E.~D. Dolan and J.~J. Mor{\'{e}}.
\newblock {Benchmarking Optimization Software with Performance Profiles}.
\newblock {\em Mathematical Programming}, 91(2):201--213, 2002.

\bibitem{DBLP:conf/iccad/DuttT97}
S.~Dutt and H.~Theny.
\newblock {Partitioning Around Roadblocks: Tackling Constraints with
  Intermediate Relaxations}.
\newblock In {\em {International Conference on Computer-Aided Design (ICCAD)}},
  pages 350--355, 11 1997.

\bibitem{FiducciaM82}
C.~M. Fiduccia and R.~M. Mattheyses.
\newblock {A Linear-Time Heuristic for Improving Network Partitions}.
\newblock In {\em 19th Conference on Design Automation (DAC)}, pages 175--181,
  1982.

\bibitem{garey1979computers}
M.~R. Garey and D.~S. Johnson.
\newblock {\em Computers and Intractability: A Guide to the Theory of
  NP-Completeness}, volume 174.
\newblock W.H. Freeman, San Francisco, 1979.

\bibitem{KAHYPAR-HFC}
Lars Gottesb{\"u}ren, Michael Hamann, Sebastian Schlag, and Dorothea Wagner.
\newblock {Advanced Flow-Based Multilevel Hypergraph Partitioning}.
\newblock In {\em 18th International Symposium on Experimental Algorithms (SEA
  2020)}. Schloss Dagstuhl-Leibniz-Zentrum f{\"u}r Informatik, 2020.

\bibitem{graham1969bounds}
R.~L. Graham.
\newblock {Bounds on Multiprocessing Timing Anomalies}.
\newblock {\em SIAM Journal on Applied Mathematics}, 17(2):416--429, 1969.

\bibitem{graham1979optimization}
R.~L. Graham, E.~L. Lawler, J.~K. Lenstra, and R.~Kan.
\newblock {Optimization and Approximation in Deterministic Sequencing and
  Scheduling: A Survey}.
\newblock In {\em Annals of Discrete Mathematics}, volume~5, pages 287--326.
  Elsevier, 1979.

\bibitem{Hauck:1996:MS:238604}
S.~A. Hauck.
\newblock {\em {Multi-FPGA Systems}}.
\newblock PhD thesis, 1995.

\bibitem{HeuerBA}
T.~Heuer.
\newblock {Engineering Initial Partitioning Algorithms for direct $k$-way
  Hypergraph Partitioning}.
\newblock Bachelor thesis, Karlsruhe Institute of Technology, 08 2015.

\bibitem{KaHyPar-MF-JEA}
T.~Heuer, P.~Sanders, and S.~Schlag.
\newblock {Network Flow-Based Refinement for Multilevel Hypergraph
  Partitioning}.
\newblock {\em {ACM} Journal of Experimental Algorithmics (JEA)},
  24(1):2.3:1--2.3:36, 09 2019.

\bibitem{KaHyPar-CA}
T.~Heuer and S.~Schlag.
\newblock {Improving Coarsening Schemes for Hypergraph Partitioning by
  Exploiting Community Structure}.
\newblock In {\em 16th International Symposium on Experimental Algorithms
  (SEA)}, Leibniz International Proceedings in Informatics (LIPIcs), pages
  21:1--21:19. Schloss Dagstuhl -- Leibniz-Zentrum f{\"u}r Informatik, 06 2017.

\bibitem{metismanual}
G.~Karypis.
\newblock {A Software Package for Partitioning Unstructured Graphs,
  Partitioning Meshes, and Computing Fill-Reducing Orderings of Sparse
  Matrices, Version 5.1.0}.
\newblock Technical report, University of Minnesota, 2013.

\bibitem{DBLP:conf/dac/KarypisAKS97}
G.~Karypis, R.~Aggarwal, V.~Kumar, and S.~Shekhar.
\newblock {Multilevel Hypergraph Partitioning: Application in VLSI Domain}.
\newblock In {\em 34th Conference on Design Automation (DAC)}, pages 526--529,
  6 1997.

\bibitem{DBLP:journals/vlsi/KarypisK00}
G.~Karypis and V.~Kumar.
\newblock {Multilevel \emph{k}-way Hypergraph Partitioning}.
\newblock {\em {VLSI} Design}, (3):285--300, 2000.

\bibitem{KLAlgorithm}
B.~W. Kernighan and S.~Lin.
\newblock {An Efficient Heuristic Procedure for Partitioning Graphs}.
\newblock {\em The Bell System Technical Journal}, 49(2):291--307, 2 1970.

\bibitem{Lengauer:1990}
T.~Lengauer.
\newblock {\em {Combinatorial Algorithms for Integrated Circuit Layout}}.
\newblock John Wiley \& Sons, Inc., 1990.

\bibitem{SNAP}
J.~Leskovec and A.~Krevl.
\newblock {SNAP Datasets: Stanford Large Network Dataset Collection}.
\newblock http://snap.stanford.edu/data, 2014.

\bibitem{MaasBA}
N.~Maas.
\newblock {Multilevel Hypergraph Partitioning with Vertex Weights Revisited}.
\newblock Bachelor thesis, Karlsruhe Institute of Technology, 05 2020.

\bibitem{DPapa2007}
D.~A. Papa and I.~L. Markov.
\newblock {Hypergraph Partitioning and Clustering}.
\newblock In {\em Handbook of Approximation Algorithms and Metaheuristics}.
  2007.

\bibitem{pinedo2012scheduling}
M.~Pinedo.
\newblock {\em Scheduling}, volume~29.
\newblock Springer, 2012.

\bibitem{KaHyPar-R}
S.~Schlag, V.~Henne, T.~Heuer, H.~Meyerhenke, P.~Sanders, and C.~Schulz.
\newblock {$k$-way Hypergraph Partitioning via $n$-Level Recursive Bisection}.
\newblock In {\em 18th Workshop on Algorithm Engineering and Experiments
  (ALENEX)}, pages 53--67. SIAM, 01 2016.

\bibitem{KAHYPAR-DIS}
Sebastian Schlag.
\newblock {\em {High-Quality Hypergraph Partitioning}}.
\newblock PhD thesis, Karlsruhe Institute of Technology, 2020.

\bibitem{Schulz2013_1000035713}
C.~Schulz.
\newblock {\em {High Quality Graph Partitioning}}.
\newblock PhD thesis, Karlsruhe Institute of Technology, 2013.

\bibitem{DBLP:journals/tvlsi/ShinK93}
H.~Shin and C.~Kim.
\newblock {A Simple Yet Effective Technique for Partitioning}.
\newblock {\em IEEE Transactions on Very Large Scale Integration (VLSI)
  Systems}, 1(3):380--386, 1993.

\bibitem{DBLP:journals/siamrev/VastenhouwB05}
B.~Vastenhouw and R.~H. Bisseling.
\newblock {A Two-Dimensional Data Distribution Method for Parallel Sparse
  Matrix-Vector Multiplication}.
\newblock {\em {SIAM} Review}, 47(1):67--95, 2005.

\bibitem{DAC2012}
N.~Viswanathan, C.~J. Alpert, C.~C.~N. Sze, Z.~Li, and Y.~Wei.
\newblock {The DAC 2012 Routability-Driven Placement Contest and Benchmark
  Suite}.
\newblock In {\em 49th Conference on Design Automation (DAC)}, pages 774--782.
  ACM, 6 2012.

\end{thebibliography}

\appendix

\section{Proof of Lemma~\ref{lemma:worstcaseaf}}
\label{sec:af_proof}

\worstcaseaf*

\begin{proof}
We define $\DefPrepacking{k} := \{P_1, \ldots, P_k\}$ and $\DefPartition{k} := \{V_1,\ldots,V_k\}$.
Let assume that the \lpt~algorithm assigned the $i$-th vertex $v_i$ of $\nonprepackedvertices$
to block $V_j \in \DefPartition{k}$. We define $\subpartition{j}{i}$ as a subset of block $V_j$ that only contains vertices of
$\langle v_1, \ldots, v_i \rangle \subseteq \nonprepackedvertices$ and $\fixedvertices$. Since the \lpt~algorithm always assigns
an vertex to a block with the smallest weight (see Section~\ref{sec:related_work}), the weight of $\subpartition{j}{i - 1}$
must be smaller or equal to $\frac{1}{k}(c(\fixedvertices) + \sum_{j = 1}^{i - 1} c(v_j))$ (average weight of all previously
assigned vertices), otherwise $\subpartition{j}{i - 1}$ would be not the block with the smallest weight.
\begin{align*}
\Rightarrow c(\subpartition{j}{i}) = c(\subpartition{j}{i-1}) + c(v_i)
  \le \frac{1}{k}(c(\fixedvertices) + \sum_{j = 1}^{i - 1} c(v_j)) + c(v_i) \le \frac{1}{k}c(\fixedvertices) + \afdbound{k}{\nonprepackedvertices}
\end{align*}
We can establish an upper bound on the weight of all blocks to which the \lpt~algorithm assigns an vertex to
with $\frac{1}{k}c(\fixedvertices) + \afdbound{k}{\nonprepackedvertices}$.
If the \lpt~algorithm does not assign any vertex to a block $V_j \in \DefPartition{k}$, its weight is equal to $c(P_j) \le \max(\DefPrepacking{k})$.
\[\Rightarrow \max(\DefPartition{k}) \le \max\{\frac{1}{k}c(\fixedvertices) + \afdbound{k}{\nonprepackedvertices}, \max(\DefPrepacking{k}) \}\]
\end{proof}

\section{Generalized Balance Property}
\label{sec:generalized_balance_property}

\begin{definition}{\definitionname{Generalized Balance Property}}
  \label{def:generalized_balance_property}
  Let $H = (V,E,c,\omega)$ be a hypergraph for which we want to compute an $\varepsilon$-balanced $k$-way
  partition and $\Prepacking := \{P_1, P_2\}$ be a prepacking of $H$ for a set of fixed vertices
  $\fixedvertices \subseteq V$. Furthermore, let  $\heaviestvertices{t_1}$ resp.~$\heaviestvertices{t_2}$ be the sequence of
  the $t_1$ resp.~$t_2$ heaviest ordinary vertices of $V \setminus \fixedvertices$ sorted in decreasing
  vertex weight order such that $t_1$ resp.~$t_2$ is the smallest number that satisfies
  $c(P_1) + c(\heaviestvertices{t_1}) \ge \balancedconstraint{2}$ resp. $c(P_2) + c(\heaviestvertices{t_2}) \ge \balancedconstraint{2}$
  (see Line~\ref{algo:balance_constraint}, Algorithm~\ref{algo:recursive_bipartitioning}).
  We say that a prepacking $\Prepacking$ satisfies the balance property with respect to $k$ if the following conditions hold:
  \begin{enumerate}
    \item[(i)] \Prepacking~is deeply balanced
    \item[(ii)] $\frac{1}{k_1}c(P_1) + \afdbound{k_1}{\heaviestvertices{t_1}} \le \balancedconstraint{k}$ with $k_1 := \lceil \frac{k}{2} \rceil$
    \item[(iii)] $\frac{1}{k_2}c(P_2) + \afdbound{k_2}{\heaviestvertices{t_2}} \le \balancedconstraint{k}$ with $k_2 := \lfloor \frac{k}{2} \rfloor$
  \end{enumerate}
\end{definition}

The proof of Theorem~\ref{theorem:sufficiently_balanced} can be adapted such that we show that there exist a $k_1$- resp.~ $k_2$-way partition
$\DefPartition{k_1}$ resp.~$\DefPartition{k_2}$ for $V_1$ resp.~$V_2$ of any balanced bipartition $\DefPartition{2} := \{V_1,V_2\}$ that respects
the prepacking \Prepacking~with $\max(\DefPartition{k_1}) \le \frac{1}{k_1}c(P_1) + \afdbound{k_1}{\heaviestvertices{t_1}} \le \balancedconstraint{k}$ (Defintion (ii))
and $\max(\DefPartition{k_2}) \le \frac{1}{k_2}c(P_2) + \afdbound{k_2}{\heaviestvertices{t_2}} \le \balancedconstraint{k}$ (Defintion (iii)).

\section{Proof of Claim~\ref{claim:balancepropertyclaim}}
\label{sec:heaviest_elements_proof}

\begin{lemma}
  \label{lemma:boundheaviestelements}
  Let $L = \langle a_1, \ldots, a_n \rangle$ be a sequence of elements sorted in decreasing weight order
  with respect to a weight function $\loadbalancingweight: L \rightarrow \mathbb{R}_{\ge 0}$
  (for a subsequence $A := \langle a_1, \ldots, a_l \rangle$ of $L$, we define $c(A) := \sum_{i = 1}^{l} c(a_i)$),
  $L'$ be an abitrary subsequence of $L$ sorted in decreasing weight order and
  $L_m = \langle a_1, \ldots, a_m \rangle$ the subsequence of the $m \le n$ heaviest elements in $L$.
  Then the following conditions hold:
  \begin{enumerate}
    \item[(i)] If $\loadbalancingweight(L') \le \loadbalancingweight(L_m)$, then $\afdbound{k}{L'} \le \afdbound{k}{L_m}$
    \item[(ii)] If $\loadbalancingweight(L') > \loadbalancingweight(L_m)$, then $\afdbound{k}{L'} - \frac{1}{k}\loadbalancingweight(L') \le \afdbound{k}{L_m} - \frac{1}{k}\loadbalancingweight(L_m) $
  \end{enumerate}
\end{lemma}

\begin{proof}
For convenience, we define $L' := \langle b_1, \ldots, b_l \rangle$.
Note that $\forall i \in \{1, \ldots, \min(m,l) \}: \loadbalancingweight(a_i) \ge \loadbalancingweight(b_i)$,
since $L_m$ contains the $m$ heaviest elements in decreasing order.
We define $i := \arg\max_{i \in \{1,\ldots,l\}} \loadbalancingweight(b_i) + \frac{1}{k}\sum_{j = 1}^{i - 1} \loadbalancingweight(b_j)$
(index that maximizes $\afdbound{k}{L'}$).

(i) + (ii): If $i \le m$, then
\[\afdbound{k}{L'} =
  \loadbalancingweight(b_i)  + \frac{1}{k} \sum_{j = 1}^{i - 1} \loadbalancingweight(b_j)
   \stackrel{\forall j \in [1,i]:~\loadbalancingweight(b_j) \le \loadbalancingweight(a_j)}{\le}
   \loadbalancingweight(a_i)  + \frac{1}{k} \sum_{j = 1}^{i - 1} \loadbalancingweight(a_j)
   \le \afdbound{k}{L_m}
\]

(i): If $m < i \le l$, then
\begin{align*}
  \afdbound{k}{L'} = \loadbalancingweight(b_i)  + \frac{1}{k} \sum_{j = 1}^{i - 1} \loadbalancingweight(b_j) =
   \loadbalancingweight(b_i) - \frac{1}{k} \sum_{j = i}^{n} \loadbalancingweight(b_j) + \frac{1}{k}\loadbalancingweight(L')
   \le \left( 1 - \frac{1}{k} \right)\loadbalancingweight(b_i) + \frac{1}{k}\loadbalancingweight(L') \\
   \stackrel{\substack{c(b_i) \le c(a_m) \\ c(L') \le c(L_m)}}{\le} \left( 1 - \frac{1}{k} \right)\loadbalancingweight(a_m) + \frac{1}{k}\loadbalancingweight(L_m)
   = \loadbalancingweight(a_m) + \frac{1}{k} \sum_{j = 1}^{m - 1} \loadbalancingweight(a_j)
   \le \afdbound{k}{L_m}
\end{align*}

(ii): If $m < i \le l$, then
\begin{align*}
  \afdbound{k}{L'} - \frac{1}{k}\loadbalancingweight(L') =
  \loadbalancingweight(b_i)  + \frac{1}{k} \sum_{j = 1}^{i - 1} \loadbalancingweight(b_j) - \frac{1}{k}\loadbalancingweight(L') =
  \loadbalancingweight(b_i) - \frac{1}{k} \sum_{l = i}^{n} \loadbalancingweight(b_l)
  \le \left( 1 - \frac{1}{k} \right)\loadbalancingweight(b_i) \\
  \stackrel{c(b_i) \le c(a_m)}{\le} \left( 1 - \frac{1}{k} \right)\loadbalancingweight(a_m)
  = \loadbalancingweight(a_m) + \frac{1}{k} \sum_{j = 1}^{m - 1} \loadbalancingweight(a_j) - \frac{1}{k}\loadbalancingweight(L_m)
  \le \afdbound{k}{L_m} - \frac{1}{k}\loadbalancingweight(L_m)
\end{align*}
\end{proof}

\balancepropertyclaim*

\begin{proof}
  Remember, $\Prepacking = \{P_1, P_2\}$, $\DefPartition{2} = \{V_1, V_2\}$ with $P_1 \subseteq V_1$ and $P_2 \subseteq V_2$, $\nonprepackedvertices_1$ is equal to
  $V_1 \setminus P_1$ and $\heaviestvertices{t}$ represents the $t$ heaviest vertices of $(V_1 \cup V_2) \setminus (P_1 \cup P_2)$ with
  $\max(\Prepacking) + c(\heaviestvertices{t}) \ge \balancedconstraint{2}$
  as defined in Definition~\ref{def:balance_property}.
  The following proof distingush two cases based on Lemma~\ref{lemma:boundheaviestelements}.

  If $c(\nonprepackedvertices_1) \le c(\heaviestvertices{t})$, then
  \begin{align*}
    \frac{1}{k'}c(P_1) + \afdbound{k'}{\nonprepackedvertices_1}
    \stackrel{\text{Lemma~\ref{lemma:boundheaviestelements}(i)}}{\le} \frac{1}{k'}c(P_1) + \afdbound{k'}{\heaviestvertices{t}}
    \stackrel{c(P_1) \le \max(\Prepacking)}{\le} \frac{1}{k'} \max(\Prepacking)  + \afdbound{k'}{\heaviestvertices{t}}
  \end{align*}

  If $c(\nonprepackedvertices_1) > c(\heaviestvertices{t})$, then
  \begin{align*}
    \frac{1}{k'}c(P_1) + \afdbound{k'}{\nonprepackedvertices_1}
    = \frac{1}{k'}c(P_1) + \afdbound{k'}{\nonprepackedvertices_1} - \frac{1}{k'}c(\nonprepackedvertices_1) +  \frac{1}{k'}c(\nonprepackedvertices_1) \\
    \stackrel{\text{Lemma~\ref{lemma:boundheaviestelements}(ii)}}{\le} \frac{1}{k'}(c(P_1) + c(\nonprepackedvertices_1))  + \afdbound{k'}{\heaviestvertices{t}} - \frac{1}{k'}c(\heaviestvertices{t})
    \stackrel{c(P_1) + c(\nonprepackedvertices_1) = c(V_1)}{=} \frac{1}{k'}(c(V_1) - c(\heaviestvertices{t}))  + \afdbound{k'}{\heaviestvertices{t}} \\
    \stackrel{c(V_1) \le \balancedconstraint{2}}{\le} \frac{1}{k'} ( \balancedconstraint{2} - c(\heaviestvertices{t}) ) + \afdbound{k'}{\heaviestvertices{t}}
    \stackrel{\max(\Prepacking) + c(\heaviestvertices{t}) \ge \balancedconstraint{2}}{\le} \frac{1}{k'} \max(\Prepacking)  + \afdbound{k'}{\heaviestvertices{t}}
  \end{align*}
\end{proof}

\section{Vertex Weight Distributions}
\label{sec:vertex_weight_distribution}

\begin{figure*}[!htb]
  \vspace{-1cm}
	\externalizedfigure{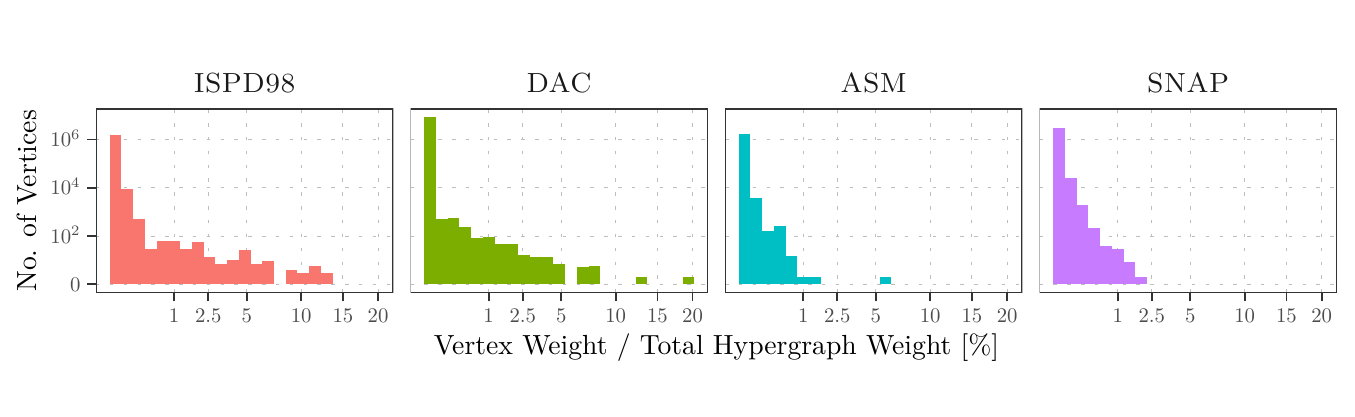}{stats/instance_weight_distributions}
	\vspace{-0.3cm}
  \caption{Overview of the vertex weight distribution for each instance type.
           The histograms (bin width $= 0.2$) show the number of vertices (y-axis) with a certain share on the total
           weight of its correponding hypergraph (x-axis). The share of a vertex $v \in V$ on the total weight of
           a weighted hypergraph $H = (V,E,c,\omega)$ is $c(v) / c(V)$.}
  \label{fig:vertex_weights}
	\vspace{-0.35cm}
\end{figure*}

\section{Configuration of Evaluated Partitioners}
\label{sec:evaluated_partitioner_config}

\Partitioner{hMetis} does not directly optimize the $(\lambda - 1)$-metric. Instead it optimizes the \emph{sum-of-external-degrees} (SOED),
which is closely related to the connectivity metric: $(\lambda - 1)(\Partition) = \mathrm{SOED}(\Partition) - \mathrm{cut}(\Partition)$.
We therefore configure \Partitioner{hMetis} to optimize SOED and calculate the $(\lambda - 1)$-metric accordingly.
The same approach is also used by the authors of \Partitioner{hMetis}~\cite{DBLP:journals/vlsi/KarypisK00}.
Additionally, \Partitioner{hMetis-R} defines the maximum allowed imbalance of a partition differently~\cite{metismanual}.
For example, an imbalance value of 5 means that a block weight between $0.45 \cdot c(V)$ and $0.55 \cdot c(V)$ is allowed at each bisection step.
We therefore translate the imbalance parameter $\varepsilon$ to a modified parameter $\varepsilon'$ such that the correct allowed
block weight is matched after $\log_2(k)$ bisections:
\[\varepsilon' := 100 \cdot \left( \left( (1 + \varepsilon) \frac{\lceil \frac{c(V)}{k} \rceil}{c(V)} \right)^{\frac{1}{\log_2(k)}} - 0.5 \right)\]

\Partitioner{PaToH} is evaluated with both the default (\Partitioner{PaToH-D}) and the quality preset (\Partitioner{PaToH-Q}).
However, there are also more fine-grained parameters available for \Partitioner{PaToH} as described in~\cite{PaToHManual}.
In our case, the \texttt{balance} parameter is of special interest as it might affect the ability of \Partitioner{PaToH} to find a balanced partition.
Therefore, we evaluated the performance of \Partitioner{PaToH} on our benchmark set with each of the possible options \texttt{Strict}, \texttt{Adaptive} and \texttt{Relaxed}.
The configuration using the \texttt{Strict} option (which is also the default) consistently produced fewest imbalanced partitions and had similar quality to the other configurations.
Consequently, we only report the results of this configuration.

\section{Number of Imbalanced Partitions per $k$ and $\varepsilon$}
\label{sec:imbalance_per_type}

\begin{table}[H]
  \centering
  \caption{Percentage of imbalanced instances produced by each partitioner on our \InstanceType{RealWorld} benchmark set for each combination of $k$ and $\varepsilon$.}
  \label{table:partitioner_stats}
  \begin{tabular}{r|rrrrrrrrr|rrr}
  %                                          & \multicolumn{12}{c}{Number of Imbalanced Partitions [\%]} \\
                                            & \multicolumn{3}{c}{\footnotesize{$k \in \{2,4,8\}$}} &
                                              \multicolumn{3}{c}{\footnotesize{$k \in \{16, 32\}$}} &
                                              \multicolumn{3}{c|}{\footnotesize{$k \in \{64,128\}$}} &
                                              \multicolumn{3}{c}{\footnotesize{Total [\%]}} \\
  $\varepsilon$                             & \scriptsize{0.01} & \scriptsize{0.03}  & \scriptsize{0.1} &
                                              \scriptsize{0.01} & \scriptsize{0.03}  & \scriptsize{0.1} &
                                              \scriptsize{0.01} & \scriptsize{0.03}  & \scriptsize{0.1} &
                                              \scriptsize{0.01} & \scriptsize{0.03}  & \scriptsize{0.1} \\
  \midrule
  	% Columns: k_2_4_8, k_16_32, k_64_128,  
 \Partitioner{\Partitioner{KaHyPar-BP-K}} & 0.0 & 0.0 & 0.0 & 0.0 & 0.0 & 0.0 & 0.0 & 0.0 & 0.0 & 0.0 & 0.0 & 0.0 \tabularnewline
 \Partitioner{\Partitioner{KaHyPar-BP-R}} & 0.0 & 0.0 & 0.0 & 0.0 & 0.0 & 0.0 & 0.0 & 0.0 & 0.0 & 0.0 & 0.0 & 0.0 \tabularnewline
 \Partitioner{\Partitioner{KaHyPar-K}} & 0.7 & 0.0 & 0.0 & 5.0 & 6.0 & 2.0 & 11.1 & 9.1 & 4.0 & 4.9 & 4.3 & 1.7 \tabularnewline
 \Partitioner{\Partitioner{KaHyPar-R}} & 2.0 & 0.7 & 0.7 & 11.0 & 9.0 & 7.0 & 21.0 & 18.0 & 13.0 & 10.0 & 8.0 & 6.0 \tabularnewline
 \Partitioner{\Partitioner{hMetis-K}} & 12.0 & 2.0 & 0.0 & 53.0 & 21.0 & 11.0 & 76.0 & 57.0 & 14.0 & 42.0 & 23.1 & 7.1 \tabularnewline
 \Partitioner{\Partitioner{hMetis-R}} & 2.7 & 2.0 & 0.0 & 18.0 & 14.0 & 7.0 & 34.0 & 27.0 & 17.0 & 16.0 & 12.6 & 6.9 \tabularnewline
 \Partitioner{\Partitioner{PaToH-Q}} & 15.3 & 2.7 & 0.7 & 28.0 & 11.0 & 5.0 & 34.0 & 22.0 & 11.0 & 24.3 & 10.6 & 4.9 \tabularnewline
 \Partitioner{\Partitioner{PaToH-D}} & 9.3 & 2.7 & 0.7 & 18.0 & 11.0 & 4.0 & 32.0 & 22.0 & 10.0 & 18.3 & 10.6 & 4.3 \tabularnewline

  \end{tabular}
\end{table}

\section{Prepacking Algorithm Statistics for \Partitioner{KaHyPar-BP-K}}
\label{sec:prepacking_stats_kahypar_bp_k}

\begin{table}[H]
	\centering
	\caption{  Occurrence of prepacked vertices (i.e., vertices that are fixed to a specific block during partitioning) for each combination of $k$ and $\varepsilon$ when using \Partitioner{KaHyPar-BP-K} on \InstanceType{RealWorld} instances:
  Minimum/average/maximum percentage of prepacked vertices (left), and percentage of instances for which the prepacking is executed at least once (right).}
    \label{table:prepacking_stats_kahypar_bp_k}
	\begin{tabular}{r|rrrrrrrrr||rrr}
		&  \multicolumn{3}{c}{$\varepsilon = 0.01$} & \multicolumn{3}{c}{$\varepsilon = 0.03$} & \multicolumn{3}{c||}{$\varepsilon = 0.1$} & \multicolumn{3}{c}{Prepacking Triggered [\%]} \\
		$k$  &   \scriptsize{Min} & \scriptsize{Avg} & \scriptsize{Max} &
		\scriptsize{Min} & \scriptsize{Avg} & \scriptsize{Max} &
		\scriptsize{Min} & \scriptsize{Avg} & \scriptsize{Max} &
		\scriptsize{$\varepsilon = 0.01$} & \scriptsize{$\varepsilon = 0.03$} & \scriptsize{$\varepsilon = 0.1$} \\
		\midrule
		2 & - & - & - & - & - & - & - & - & - & - & - & - \tabularnewline
4 & - & - & - & - & - & - & - & - & - & - & - & - \tabularnewline
8 & 6.7 & 17.1 & 41.6 & 0.4 & 0.5 & 0.6 & 2.3 & 2.3 & 2.3 & 5.0 & 3.3 & 1.7 \tabularnewline
16 & 3.1 & 15.6 & 34.0 & 0.2 & 2.0 & 7.2 & 1.9 & 2.1 & 2.3 & 8.3 & 6.7 & 3.3 \tabularnewline
32 & 0.3 & 29.9 & 56.0 & 0.1 & 11.7 & 42.3 & 0.2 & 3.4 & 26.3 & 13.3 & 15.0 & 6.7 \tabularnewline
64 & 0.2 & 54.4 & 94.3 & 0.3 & 23.0 & 69.3 & 0.4 & 6.6 & 94.7 & 21.7 & 10.0 & 8.3 \tabularnewline
128 & 0.5 & 76.5 & 100.0 & 0.4 & 42.4 & 91.0 & 0.3 & 15.7 & 59.8 & 28.3 & 21.7 & 11.7 \tabularnewline

	\end{tabular}
\end{table}

\clearpage

\section{Quality Comparison for $\varepsilon = 0.03$ and $\varepsilon = 0.1$}
\label{sec:quality_comparison}

\begin{figure*}[!htb]
	\centering
	\begin{minipage}{.49\textwidth}
		\externalizedfigure{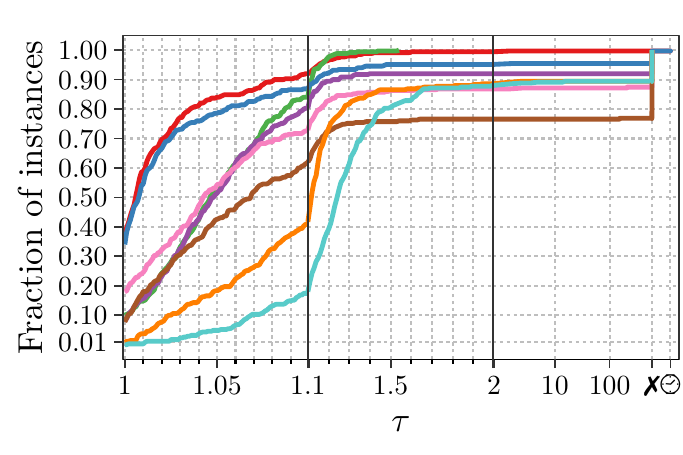}{plots/realworld_e0.03}
	\end{minipage} %
	\begin{minipage}{.49\textwidth}
		\externalizedfigure{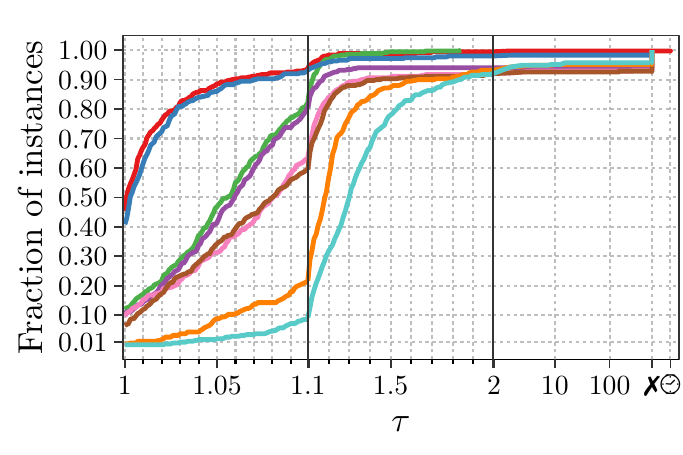}{plots/realworld_e0.1}
	\end{minipage} %
	\vfill
	\vspace{-0.4cm}
	\begin{minipage}{\textwidth}
		\externalizedfigure{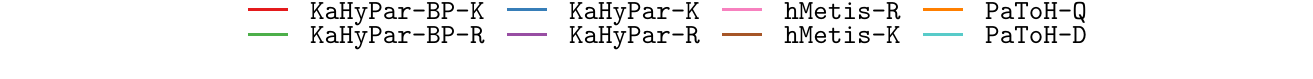}{plots/legend}
	\end{minipage} %
	\vspace{-0.4cm}
	\caption{Comparing the solution quality of each evaluated partitioner for $\varepsilon = 0.03$ (left) and
 	         $\varepsilon = 0.1$ (right) on our \InstanceType{RealWorld} benchmark set.
 	         Note, \ClockLogo~marks instances that exceeded the time limit.}
	\label{fig:quality_epsilon_real_world}
\end{figure*}

\begin{figure*}[!htb]
	\centering
	\begin{minipage}{.49\textwidth}
		\externalizedfigure{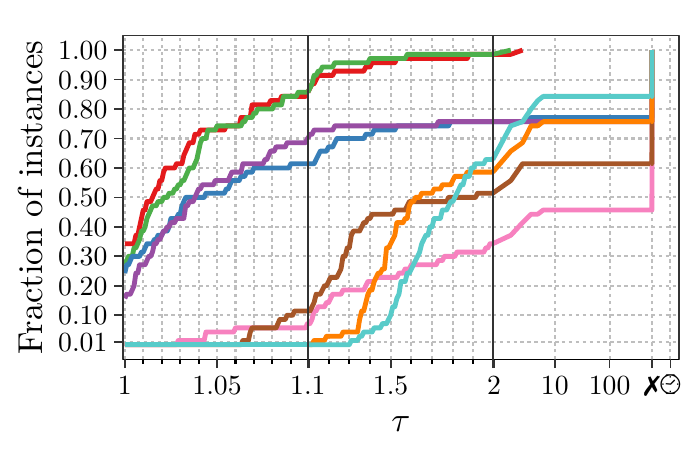}{plots/artificial_e0.03}
	\end{minipage} %
	\begin{minipage}{.49\textwidth}
		\externalizedfigure{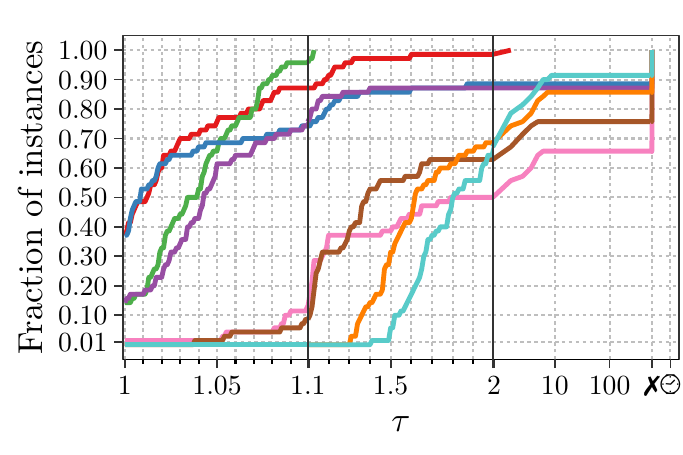}{plots/artificial_e0.1}
	\end{minipage} %
	\vfill
	\vspace{-0.4cm}
	\begin{minipage}{\textwidth}
		\externalizedfigure{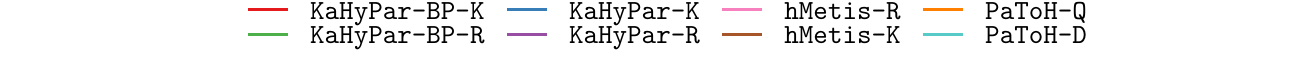}{plots/legend}
	\end{minipage} %
	\vspace{-0.25cm}
	\caption{Comparing the solution quality of each evaluated partitioner for $\varepsilon = 0.03$ (left) and
             $\varepsilon = 0.1$ (right) on our \InstanceType{Artificial} benchmark set.
             Note, \ClockLogo~marks instances that exceeded the time limit.}
	\label{fig:quality_epsilon_artificial}
\end{figure*}

\clearpage

\section{Absolute Running Times}
\label{sec:running_time}

\begin{figure*}[!ht]
	\externalizedfigure{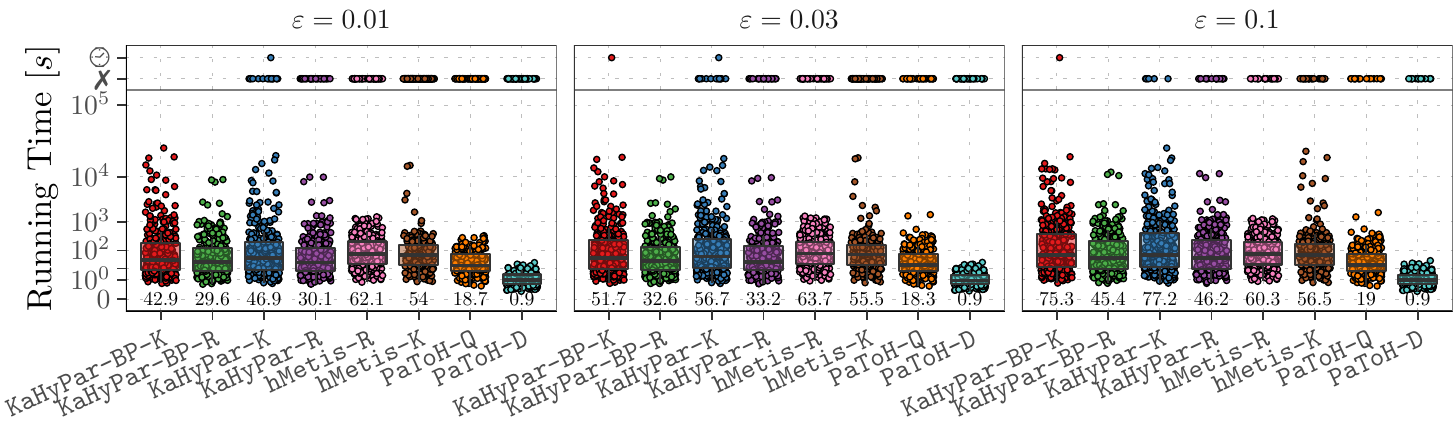}{plots/running_time_realworld}
	\vspace{-0.5cm}
  \caption{Comparing the running time of each evaluated partitioner for different values of $\varepsilon$ on our \InstanceType{RealWorld} benchmark set.
           The number under each boxplot denotes the average running time of the corresponding partitioner.
           Note, \ClockLogo~marks instances that exceeded the time limit.}
  \label{fig:running_time_real_world}
\end{figure*}

\begin{figure*}[!ht]
	\centering
	\externalizedfigure{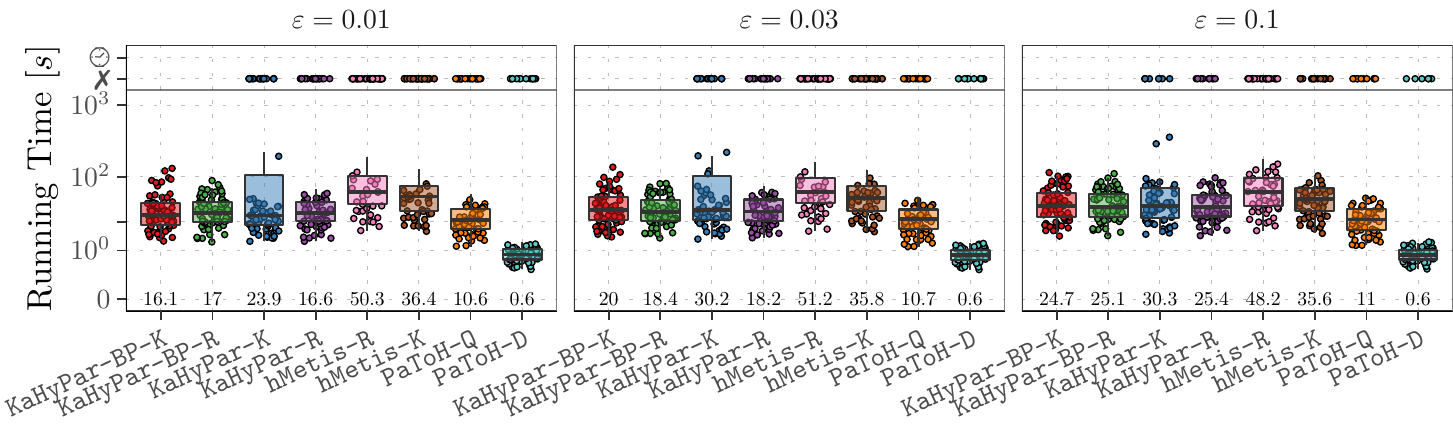}{plots/running_time_artificial}
	\vspace{-0.5cm}
  \caption{Comparing the running time of each evaluated partitioner for different values of $\varepsilon$ on our \InstanceType{Artificial} benchmark set.
           The number under each boxplot denotes the average running time of the corresponding partitioner.
           Note, \ClockLogo~marks instances that exceeded the time limit.}
  \label{fig:running_time_artificial}
\end{figure*}

\end{document}